\documentclass[11pt]{article}

\usepackage{enumerate}
\usepackage[english]{babel}

\usepackage{amsmath,amsfonts,amssymb,amsthm}
\usepackage{bbm}
\usepackage{hyperref}
\usepackage{color}
\usepackage[margin=0.9in]{geometry}

\usepackage{appendix}
\usepackage{graphicx}
\usepackage{subcaption}

\newtheorem{theorem}{Theorem}[section]

\newtheorem*{remark*}{Remark}
\newtheorem*{notation*}{Notation}
\newtheorem*{observation*}{Observation}
\newtheorem*{theorem*}{Theorem}

\usepackage{tikz}
\usetikzlibrary{shapes}
\usetikzlibrary{patterns}
\usetikzlibrary{arrows,positioning,decorations.pathmorphing,trees}
\usepackage{pgfplots}

\tikzstyle{n}=[]
\tikzstyle{u}=[circle, draw, fill=black, inner sep=0pt, minimum width=3pt]
\tikzstyle{sq}=[square, draw, fill=blue, inner sep=0pt, minimum width=3pt]

\usepackage[numbers]{natbib}

\usepackage{multirow}
\usepackage{array}

\title{The Declining Price Anomaly is not Universal\\[.2cm] in Multi-Buyer Sequential Auctions (but almost is)}

\author{V.V. Narayan\thanks{McGill University: {\tt vishnu.narayan@mail.mcgill.ca}}
\and E. Prebet\thanks{\'{E}cole polytechnique f\'{e}d\'{e}rale de Lausanne: {\tt enguerrand.prebet@epfl.ch}} 
\and A. Vetta\thanks{McGill University: {\tt adrian.vetta@mcgill.ca}}}

\begin{document}
\maketitle

\begin{abstract}
The declining price anomaly states that the price weakly decreases
when multiple copies of an item are sold sequentially over time.
The anomaly has been observed in a plethora of practical applications.
On the theoretical side, Gale and Stegeman~\cite{GS01} proved that the anomaly is guaranteed 
to hold in full information sequential auctions with exactly two buyers. 
We prove that the declining price anomaly is {\em not} guaranteed 
in full information sequential auctions with three or more buyers. 
This result applies to both first-price and second-price sequential auctions. Moreover, it applies regardless of 
the tie-breaking rule used to generate equilibria in these sequential auctions.
To prove this result we provide a refined treatment of subgame perfect equilibria that survive the iterative deletion of weakly dominated strategies
and use this framework to experimentally generate a very large number of random sequential auction instances.
In particular, our experiments produce an instance with three bidders and eight items that, for a specific tie-breaking rule, 
induces a non-monotonic price trajectory.  Theoretic analyses are then applied to show that this instance can be used to prove that
for every possible tie-breaking rule there is a sequential auction on which it induces a non-monotonic price trajectory.
On the other hand, our experiments show that non-monotonic price trajectories are extremely rare.
In over six million experiments only a $0.000183$ proportion of the instances violated the declining price anomaly.
\end{abstract}

\section{Introduction}\label{sec:intro}
In a sequential auction identical copies of an item are sold over time.
In a private values model with {\em unit-demand}, risk neutral buyers, 
Milgrom and Weber~\citep{MW00,Web83} showed that the sequence of prices
forms a martingale. In particular, expected prices are constant over time.\footnote{If the
values are affiliated then prices can have an upwards drift.}
In contrast, on attending a wine auction, 
\citet{Ash89} made the surprising observation that prices for identical lots
declined over time: ``The law of the one price was repealed and no one even seemed to notice!"
This {\em declining price anomaly} was also noted in sequential auctions for the disparate examples of 
livestock~(\citet{Buc82}), Picasso prints~(\citet{PS96}) and satellite transponder leases~(\citet{MW00}).
Indeed, the possibility of decreasing prices in a sequential auction was raised by~\citet{Sos61} nearly sixty years ago.
An assortment of reasons have been given to explain this anomaly. 
In the case of wine auctions, proposed causes include absentee buyers
utilizing non-optimal bidding strategies~(\citet{Gin98}) and the {\em buyer's option rule} where the auctioneer may
allow the buyer of the first lot to make additional purchases at the same price~(\citet{BM92}). Minor non-homogeneities 
amongst the items can also lead to falling prices. For example, in the case of art prints the items may suffer
slight imperfections or wear-and-tear; as a consequence, the auctioneer may sell the prints in decreasing order of quality~(\citet{PS96}).
More generally, a decreasing price trajectory may arise due to risk-aversion,
such as non-decreasing, absolute risk-aversion~(\citet{MV93}) or aversion to price-risk~(\citet{Mez11}); see 
also \citet{HZ15}. Further potential economic and behavioural explanations have been provided in \citep{Gin98, Ashta06, Tu10}.

Of course, most of these explanations are context-specific. However, the declining price anomaly appears more universal.
In fact, in practice the anomaly is ubiquitous: It has now been observed in sequential auctions
for antiques~(\citet{GO07}), commercial real estate~(\citet{Lus94}), condominiums~(\citet{AG92}), fish~(\citet{GGK11}), flowers~(\citet{BOP01}), 
fur~(\citet{LT06}), lobsters~(\citet{SGL17}), jewellery~(\citet{CGV96}), paintings~(\citet{BG97}), stamps~(\citet{TP95}) and wool~(\citet{Bur85}). 

Given the plethora of examples, the question arises as whether this property is actually an anomaly.
In groundbreaking work, \citet{GS01} proved that it is not in sequential auctions with {\em two} bidders. Specifically, in {\em second-price}
sequential auctions with two multiunit-demand buyers, prices are weakly decreasing over time at the unique subgame perfect 
equilibrium that survives the iterative deletion of weakly dominated strategies. Moreover, this result applies regardless of the 
valuation functions of the buyers; the result also extends to the corresponding equilibrium in {\em first-price} sequential auctions.
 It is worth highlighting here two important aspects of the model studied in~\citep{GS01}.
 First, \citeauthor{GS01} consider multiunit-demand buyers whereas prior theoretical work had focussed on the 
 simpler setting of unit-demand buyers. As well as being of more practical relevance (see the many examples 
 above), multiunit-demand buyers can implement more sophisticated bidding strategies. Therefore, it is not unreasonable that
equilibria in multiunit-demand setting may possess more interesting properties than equilibria in the unit-demand setting.
 Second, they study an auction with {\em full information}.
The restriction to full information is extremely useful here as it separates away informational aspects. Hence,
it allows one to focus on the strategic properties caused purely by the sequential sales of items and not by a lack of information.

\subsection{Results and Overview of the Paper}
The result of \citet{GS01} prompts the question of whether or not the declining price anomaly is guaranteed to hold in
general, that is, in sequential auctions with more than two buyers. We answer this question in the negative
by exhibiting a sequential auction with three buyers and eight items where prices initially rise and then fall.
In order to run our experiments that find this counter-example (to the conjecture that prices are weakly decreasing for multi-buyer 
sequential auctions) we study in detail the form of equilibria in sequential auctions.
First, it is important to note that there is a fundamental 
distinction between sequential auctions with two buyers
and sequential auctions with three or more buyers. In the former sequential auction, each subgame reduces to
a standard {\em auction with independent valuations}. We explain this in Section~\ref{sec:2-buyer}, where we 
present the two-buyer full information model of~\citet{GS01}.
In contrast, in a multi-buyer sequential auction each subgame
reduces to an {\em auction with interdependent valuations}. This is explained in Section~\ref{sec:multi-buyer} after  
we present the extension of the model of \citep{GS01} to multi-buyer sequential auctions.
Consequently to study multi-buyer sequential auctions we must study the equilibria of 
auctions with interdependent valuations. A theory of such equilibria was recently developed by \citet{PST12}
via a correspondence with an ascending price mechanism. In particular, as we discuss in 
Section~\ref{sec:equilibria-interdependent}, the ascending price mechanism outputs a unique bid value, 
called the {\em dropout bid} $\beta_i$, for each buyer~$i$. 
For first-price auctions it is known \citep{PST12} that these dropout bids form a subgame perfect equilibrium 
and, moreover, the interval $[0,\beta_i]$ is the exact set of bids that survives {\em all} 
processes consisting of the iterative deletion of strategies that are weakly dominated. 
In contrast, we show in Appendix A that for second-price auctions it may be the case that no bids survive the 
iterative deletion of weakly dominated strategies; however, we prove in Section~\ref{sec:equilibria-interdependent} that the interval
$[0,\beta_i]$ is the exact set of bids for any losing buyer that survives {\em all} 
processes consisting of the iterative deletion of strategies that are weakly dominated {\em by a lower bid}. 

In Section~\ref{sec:counter-example} we describe the counter-example.
We emphasize that there is nothing unusual about our example. The form of the valuation functions
used for the buyers are standard, namely, weakly decreasing marginal valuations. Furthermore, 
the non-monotonic price trajectory does not arise because of the use of an artificial   
tie-breaking rule; the three most natural tie-breaking rules, see Section~\ref{sec:equilibria-sequential},
all induce the same non-monotonic price trajectory. Indeed, we present an even stronger result in Section \ref{sec:general}:
for {\bf any} tie-breaking rule, there is a sequential auction
on which it induces a non-monotonic price trajectory. 

This lack of weakly decreasing prices provides an explanation for why multi-buyer sequential auctions have
been hard to analyze quantitatively. We provide a second explanation in Section~\ref{sec:negative}. There we present
a three-buyer sequential auction that does satisfy weakly decreasing prices but which has subgames where 
some agent has a negative value from winning against one of the two other agents. Again, this contrasts with
the two-buyer case where every agent always has a non-negative value from winning against the other agent
in every subgame.

Finally in Section~\ref{sec:expts}, we describe the results obtained via our large scale experimentations. 
These results show that whilst the declining price anomaly is not universal, exceptions are extremely rare.
Specifically, from a randomly generated dataset of over six million sequential auctions only a $0.000183$ proportion of 
the instances produced non-monotonic price trajectories. Consequently, these experiments are consistent with the practical 
examples discussed in the introduction.
Of course, it is perhaps unreasonable to assume that subgame equilibria arise in practice;
we remark, though, that the use of simple bidding algorithms by bidders may also lead to weakly decreasing prices in a multi-buyer sequential auction.
For example, \citet{Rod09} presents a method called the {\em residual monopsonist procedure} inducing this property.

\section{The Sequential Auction Model}\label{sec:model}
Here we present the full information sequential auction model. There are $T$ identical items and $n$~buyers. 
Exactly one item is sold in each time period over $T$ time periods. Buyer~$i$ has a value $V_i(k)$ 
for winning exactly $k$ items. Thus $V_i(k)= \sum_{\ell=1}^k v_i(\ell)$, where $v_i(\ell)$ is the 
marginal value buyer~$i$ has for obtaining an $\ell$th item. This induces an extensive form game.
To analyze this game it is informative to begin by considering the $2$-buyer case, 
as studied by Gale and Stegeman~\cite{GS01}.

\subsection{The Two-Buyer Case}\label{sec:2-buyer}
During the auction, the relevant history is the number of items each buyer has currently won. 
Thus we may compactly represent the extensive form (``tree") of the auction using a directed graph with a node $(x_1,x_2)$ 
for any pair of non-negative integers that satisfies
$x_1+x_2\le T$. The node $(x_1,x_2)$ induces a subgame with $T-x_1-x_2$ items for sale and where
each buyer~$i$ already possesses $x_i$ items.
Note there is a {\em source node}, $(0,0)$, corresponding to the whole game, and {\em sink nodes} $(x_1,x_2)$, where $x_1+x_2=T$.
The values Buyer~$1$ and Buyer~$2$ have for a sink node $(x_1,x_2)$ are $\Pi_1(x_1,x_2)=V_1(x_1)$ and $\Pi_2(x_1,x_2)=V_2(x_2)$, respectively.
We want to evaluate the values (utilities) at the source node $(0,0)$.
We can do this recursively working from the sinks upwards. Take a node $(x_1,x_2)$, where $x_1+x_2=T-1$.
This node corresponds to the final round of the auction, where the last item is sold, given that each buyer~$i$ has already won $x_i$ items.
The node $(x_1,x_2)$ will have directed arcs to the sink nodes $(x_1+1,x_2)$ and $(x_1,x_2+1)$; these
arcs correspond to Buyer~$1$ and Buyer~$2$ winning the final item, respectively.
For the case of second-price auctions, it is then a weakly dominant strategy for Buyer~$1$ to bid
its marginal value $v_1(x_1+1)=V_1(x_1+1)-V_1(x_1)$; similarly for Buyer~$2$. Of course, this marginal 
value is just $v_1(x_1+1)=\Pi_1(x_1+1, x_2)-\Pi_1(x_1, x_2+1)$, the difference in value between 
winning and losing the final item.
If Buyer~$1$ is the highest bidder at $(x_1,x_2)$, that is, 
$\Pi_1(x_1+1,x_2)-\Pi_1(x_1,x_2+1)\ge \Pi_2(x_1,x_2+1)-\Pi_2(x_1+1,x_2)$, 
then we have that
\begin{eqnarray*}
\Pi_1(x_1,x_2) &=& \Pi_1(x_1+1,x_2) -\big(\, \Pi_2(x_1,x_2+1)-\Pi_2(x_1+1,x_2)\, \big) \\
\Pi_2(x_1,x_2) &=& \Pi_2(x_1+1,x_2)
\end{eqnarray*}
Symmetric formulas apply if Buyer~$2$ is the highest bidder at $(x_1,x_2)$.
Hence we may recursively define a value for each buyer for each node. The iterative elimination of
weakly dominated strategies then leads to a subgame perfect equilibrium~\cite{GS01, BBB08}.

\

\noindent{\tt Example:} Consider a two-buyer sequential auction with two items, where the marginal valuations are $\{v_1(1), v_1(2)\}=\{10,8\}$ and
$\{v_2(1), v_2(2)\}=\{6,3\}$. This game is illustrated in Figure~\ref{fig:2-buyer-second-price}. The base case with the 
values of the sink nodes is shown in Figure~\ref{fig:2-buyer-second-price}(a). The first row in each node refers to Buyer~$1$
and shows the number of items won (in plain text) and the corresponding value (in bold); the
second row refers to Buyer~$2$. The outcome of the second-price
sequential auction, solved recursively, is then shown in Figure~\ref{fig:2-buyer-second-price}(b).
Arcs are labelled by the bid value; here arcs for Buyer~$1$ point left and arcs for Buyer~$2$ point right.
Solid arcs represent winning bids and dotted arcs are losing bids. The equilibrium path is shown in bold.

 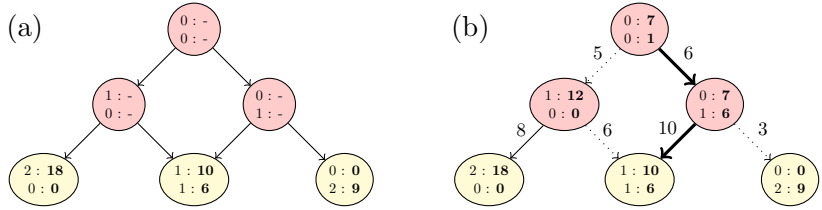
\begin{figure}[h]
\centering
 \begin{tikzpicture}[scale=0.5]
  \node (a) at (-.5,4) {(a)};
\node [ellipse,draw, fill=red!20, align=center, scale=0.5] (v1) at (4,4) {0 : - \\ 0 : - };
\node [ellipse,draw, fill=red!20, align=center, scale=0.5](v2) at (2,2) {1 : - \\ 0 : - };
\node [ellipse,draw, fill=red!20, align=center, scale=0.5](v3) at (6,2) {0 : - \\ 1 : - };
\node[ellipse,draw, fill=yellow!20,align=center, scale=0.5](v4) at (0,0) {2 : {\bf 18} \\ 0 :  {\bf 0}};
\node[ellipse,draw,fill=yellow!20, align=center, scale=0.5](v5) at (4,0) {1 :  {\bf 10} \\ 1 :  {\bf 6} };
\node[ellipse,draw, fill=yellow!20, align=center, scale=0.5](v6) at (8,0) {0 :  {\bf 0} \\ 2 :  {\bf 9} };
\draw [->] (v1) -- (v2);
\draw [->] (v1) -- (v3);
\draw [->] (v2) -- (v4);
\draw [->] (v2) -- (v5);
\draw [->] (v3) -- (v5);
\draw [->] (v3) -- (v6);
\end{tikzpicture}
 \qquad
  \begin{tikzpicture}[scale=0.5]
    \node (b) at (-.5,4) {(b)};
  \node [ellipse,draw, fill=red!20, align=center, scale=0.5] (v1) at (4,4) {0 : {\bf 7} \\ 0 : {\bf 1} };
\node [ellipse,draw, fill=red!20, align=center, scale=0.5](v2) at (2,2) {1 : {\bf 12} \\ 0 : {\bf 0} };
\node [ellipse,draw, fill=red!20, align=center, scale=0.5](v3) at (6,2) {0 : {\bf 7} \\ 1 : {\bf 6} };
\node[ellipse,draw, fill=yellow!20, align=center, scale=0.5](v4) at (0,0) {2 : {\bf 18} \\ 0 :  {\bf 0}};
\node[ellipse,draw, fill=yellow!20, align=center, scale=0.5](v5) at (4,0) {1 :  {\bf 10} \\ 1 :  {\bf 6} };
\node[ellipse,draw, fill=yellow!20, align=center, scale=0.5](v6) at (8,0) {0 :  {\bf 0} \\ 2 :  {\bf 9} };
\draw [->, dotted] (v1) -- (v2) node[very near start,left, scale = .66]{$5\ $};;
\draw [->, very thick] (v1) -- (v3) node[very near start,right, scale = .66]{$\ \ 6$};
\draw [->] (v2) -- (v4) node[very near start,left, scale = .66]{$8\ $} ;
\draw [->, dotted] (v2) -- (v5) node[very near start,right, scale = .66]{$\ 6$};
\draw [->, very thick] (v3) -- (v5) node[very near start,left, scale = .66]{$10\ $};
\draw [->, dotted] (v3) -- (v6) node[very near start,right, scale = .66]{$\ \ 3$};
\end{tikzpicture}
\caption{Second-Price Sequential Auction}
\label{fig:2-buyer-second-price}
\end{figure}
Observe that the {\em declining price anomaly} is exhibited in this example. 
Specifically, in this subgame perfect equilibrium, Buyer~$2$ wins the first item for a price $5$ and
Buyer~$1$ wins the second item for a price $3$. 
As stated, this example is not an exception. Gale and Stegeman~\cite{GS01} showed that 
weakly decreasing prices are a property of $2$-buyer sequential auctions. 
\begin{theorem}\cite{GS01}\label{thm:decreasing}
In a $2$-buyer second-price sequential auction there is a unique equilibrium that survives the iterative deletion of weakly dominated strategies.
Moreover, at this equilibrium prices are weakly declining.
\qed
\end{theorem}
 
We remark that the subgame perfect equilibrium that survives the iterative elimination of 
weakly dominated strategies is unique in terms of
the values at the nodes. Moreover, given a fixed tie-breaking rule, 
the subgame perfect equilibrium also has a unique equilibrium path in each subgame. 

In addition, Theorem~\ref{thm:decreasing} also applies to first-price sequential auctions. In this case, 
to ensure the existence of an equilibrium, we make the standard assumption that there is a fixed small bidding increment.
That is, for any price $p$ there is a unique maximum price smaller than $p$.
Given this, for the example above, the subgame perfect equilibrium using a first-price sequential auction is as 
shown in Figure~\ref{fig:2-buyer-first-price}. Here we use the notation $p^+$ to denote a winning bid of value equal to $p$,
and the notation $p$ to denote a losing bid equal to maximum value smaller than $p$.

 \begin{figure}[h]
\centering
  \begin{tikzpicture}[scale=0.5]
  \node [ellipse,draw,fill=red!20, align=center, scale=0.5] (v1) at (4,4) {0 : {\bf 7} \\ 0 : {\bf 1} };
\node [ellipse,draw, fill=red!20, align=center, scale=0.5](v2) at (2,2) {1 : {\bf 12} \\ 0 : {\bf 0} };
\node [ellipse,draw, fill=red!20,align=center, scale=0.5](v3) at (6,2) {0 : {\bf 7} \\ 1 : {\bf 6} };
\node[ellipse,draw, fill=yellow!20, align=center, scale=0.5](v4) at (0,0) {2 : {\bf 18} \\ 0 :  {\bf 0}};
\node[ellipse,draw, fill=yellow!20,align=center, scale=0.5](v5) at (4,0) {1 :  {\bf 10} \\ 1 :  {\bf 6} };
\node[ellipse,draw, fill=yellow!20, align=center, scale=0.5](v6) at (8,0) {0 :  {\bf 0} \\ 2 :  {\bf 9} };
\draw [->, dotted] (v1) -- (v2) node[very near start,left, scale = .66]{$5\ $};;
\draw [->, very thick] (v1) -- (v3) node[very near start,right, scale = .66]{$\ \ 5+$};
\draw [->] (v2) -- (v4) node[very near start,left, scale = .66]{$6+\ $} ;
\draw [->, dotted] (v2) -- (v5) node[very near start,right, scale = .66]{$\ 6$};
\draw [->, very thick] (v3) -- (v5) node[very near start,left, scale = .66]{$3+\ $};
\draw [->, dotted] (v3) -- (v6) node[very near start,right, scale = .66]{$\ \ 3$};
\end{tikzpicture}
\caption{First-Price Sequential Auction}
\label{fig:2-buyer-first-price}
\end{figure}
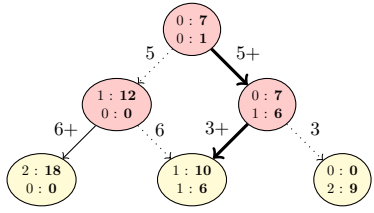

Observe that the resultant prices on the equilibrium path are more easily apparent in Figure~\ref{fig:2-buyer-first-price}
than in Figure~\ref{fig:2-buyer-second-price}. For this reason, all the figures we present in the rest of the 
paper will be for  first-price auctions; equivalent figures can can be drawn for the case of second-price auctions.

So the decreasing price anomaly holds in two-buyer sequential auctions.
The question of whether or not it applies to sequential auctions with more than two buyers remained open. 
We resolve this question in the rest of this paper. To do this, let's first study equilibria in the
full information sequential auction model when there are more than two buyers.

\subsection{The Multi-Buyer Case}\label{sec:multi-buyer}
The underlying model of~\citep{GS01} extends simply to 
sequential auctions with $n\ge 3$ buyers. There is a node $(x_1,x_2, \dots, x_n)$ for each set of
non-negative integers satisfying $ \sum_{i=1}^n x_i \le T$. There is a directed arc from $(x_1,x_2, \dots, x_n)$ to 
$(x_1,x_2, \dots, x_{j-1}, x_j+1,x_{j+1},\dots x_n)$ for each $1\le j\le n$. Thus each non-sink node has $n$ out-going arcs.
This is problematic: Whilst in the final time period each buyer has a value for winning and a value for losing, this is no longer 
the case recursively in earlier time periods. Specifically, buyer~$i$ has value for winning, but $n-1$ (different) values for losing
depending upon the identity of the buyer $j\neq i$ who actually wins. Thus rather than each node corresponding to a standard auction,
each node in the multi-buyer case corresponds to an auction with interdependent valuations. 

Formally, an {\em auction with interdependent valuations} is a single-item auction where each buyer~$i$ has a value $v_{i,i}$
for winning  the item and, for each buyer~$j\neq i$, buyer~$i$ has value $v_{i,j}$ if buyer~$j$ wins the item.
These auctions, also called {\em auctions with externalities}, were introduced by Funk~\cite{Fun96} and by Jehiel and Moldovanu~\cite{JM96}. 
Their motivations were applications where losing participants were not indifferent to the identity of the winning buyer;
examples include firms seeking to purchase a patented innovation, take-over acquisitions of a smaller company in an 
oligopolistic market, and sports teams competing to sign a star athlete. 

Therefore to understand multi-buyer sequential auctions we must first understand equilibria in auctions with interdependent valuations.
This is actually not a simple task. Indeed such an understanding was only recently provided by Paes Leme et al.~\cite{PST12}.

\subsection{Equilibria in Auctions with Interdependent Valuations}\label{sec:equilibria-interdependent}

We can explain the result of~\citep{PST12} via an ascending price auction.

\subsubsection{An Ascending Price Mechanism}\label{sec:ascending}
Imagine a two-buyer ascending price auction where the valuations of the buyers are $v_1>v_2$. 
The requested price $p$ starts at zero and continues to rise until the point where the second buyer drops out.
Of course, this happens when the price reaches $v_2$, and so Buyer~$1$ wins for a payment $p^+=v_2$.
But this is exactly the outcome expected from a first-price auction: Buyer~$2$ loses with bid of $p$ and Buyer~$1$ 
wins with a bid of $p^+$. To generalize this to multi-buyer settings we
can view this process as follows. At a price $p$, buyer~$i$ remains in the auction as long as there
is at least one buyer~$j$ {\em still in the auction} who buyer~$i$ is willing to pay a price $p$
to beat; that is, $v_{i,i}-p>v_{i,j}$. The last buyer to drop out wins at the corresponding price.
For example, in the two-buyer example above, Buyer~$2$ drops out
at price $p=v_2$ as it would rather lose to Buyer~$1$ than win above that price. Therefore, at price $p^+$ there is no
buyer still in the auction that Buyer~$1$ wishes to beat (because there are no other buyers remaining in the
auction at all!). Thus Buyer~$1$ drops out at $p^+$ and, being the last buyer to drop out, wins at that price.  

Observe that, even in the multi-buyer setting, this procedure produces a unique {\em dropout bid} $\beta_i$ for each buyer~$i$.
To illustrate this, two auctions with interdependent valuations are shown in Figure~\ref{fig:4-buyer-example}. 
In these diagrams the label of an arc from buyer~$i$ to buyer~$j$ is $w_{i,j}=v_{i,i}-v_{i,j}$. That is, buyer~$i$ is willing to pay up to
$w_{i,j}$ to win {\em if the alternative is that buyer~$j$ wins the item}. Now consider running our ascending price procedure
for these auctions. In Figure~\ref{fig:4-buyer-example}(a),
Buyer 1 drops out when the price reaches $18$. Since Buyer~$1$ is no longer active in the auction, Buyer~$4$ drops out 
when the price reaches $23$. At this point, Buyer~$2$ and Buyer~$3$ are left to compete for the item. Buyer~$3$ wins when Buyer~$2$ drops
out at price $31$. Thus the drop-out bid of Buyer~$2$ is $31^+$.
Observe that Buyer~$2$ loses despite having very high values for winning (against Buyer~$1$ and Buyer~$4$).

The example of Figure~\ref{fig:4-buyer-example}(b) with dropout bid vector $(\beta_1,\beta_2,\beta_3, \beta_4)=(24,24,24,24^+)$ 
is more subtle. Here Buyer~$2$ drops out at price $24$. But
Buyer~$3$ only wanted to beat Buyer~$2$ at this price so it then immediately drops out at the same price. Now
Buyer~$1$ only wanted to beat Buyer~$2$ and Buyer~$3$ at this price, so it then immediately drops out at the same price.
This leaves Buyer~$4$ the winner at price $24^+$.

 \begin{figure}[!h]
\centering
 \begin{tikzpicture}[scale=0.3]
  \node (Ex1) at (-12,0) {(a)};
 \node[rectangle,draw](B1) at (-6,0) {Buyer 1};
\node[rectangle,draw, scale=1](B2) at (-6,-6) {Buyer 2};
\node[rectangle,draw, scale=1](B3) at (6,-6) {Buyer 3};
\node[rectangle,draw, scale=1](B4) at (6,0) {Buyer 4};

\draw [->,  thick] (B1) -- (-6, -2.5) node[midway, left, scale = .66]{$18\ $};
\draw [->,  thick] (B1) -- (-0.5,0) node[midway, above, scale = .66]{$13\ $};
\draw [->, thick] (B1) -- (-.5, -2.5) node[midway,above, scale = .66]{$\ \ 14$};

\draw [->, thick] (B2) -- ( -6 ,-3.5) node[midway,left, scale = .66]{$97\ $};
\draw [->,  thick] (B2) -- (-.5,-6) node[midway, below, scale = .66]{$31\ $};
\draw [->,  thick] (B2) -- (-.5,-3.5) node[midway, below, scale = .66]{$74\ $};

\draw [->, thick] (B3) -- ( .5 ,-6) node[midway,below, scale = .66]{$33$};
\draw [->, thick] (B3) -- (.5,-3.5) node[midway, below, scale = .66]{$\ \ 12$};
\draw [->, thick] (B3) -- (6,-3.5) node[midway, right, scale = .66]{$\ \ 11$};

\draw [->, thick] (B4) -- ( .5 ,-2.5) node[midway,above, scale = .66]{$10$};
\draw [->, thick] (B4) -- (6,-2.5) node[midway, right, scale = .66]{$\ \ 23$};
\draw [->, thick] (B4) -- (.5,0) node[midway, above, scale = .66]{$\ \ 35$};

 \end{tikzpicture}
\quad \quad \quad 
  \begin{tikzpicture}[scale=0.3]
    \node (Ex2) at (-12,0) {(b)};
 \node[rectangle,draw](B1) at (-6,0) {Buyer 1};
\node[rectangle,draw, scale=1](B2) at (-6,-6) {Buyer 2};
\node[rectangle,draw, scale=1](B3) at (6,-6) {Buyer 3};
\node[rectangle,draw, scale=1](B4) at (6,0) {Buyer 4};

\draw [->,  thick] (B1) -- (-6, -2.5) node[midway, left, scale = .66]{$37\ $};
\draw [->,  thick] (B1) -- (-0.5,0) node[midway, above, scale = .66]{$22\ $};
\draw [->, thick] (B1) -- (-.5, -2.5) node[midway,above, scale = .66]{$\ \ 59$};

\draw [->, thick] (B2) -- ( -6 ,-3.5) node[midway,left, scale = .66]{$17\ $};
\draw [->,  thick] (B2) -- (-.5,-6) node[midway, below, scale = .66]{$13\ $};
\draw [->,  thick] (B2) -- (-.5,-3.5) node[midway, below, scale = .66]{$24\ $};

\draw [->, thick] (B3) -- ( .5 ,-6) node[midway,below, scale = .66]{$63$};
\draw [->, thick] (B3) -- (.5,-3.5) node[midway, below, scale = .66]{$\ \ 19$};
\draw [->, thick] (B3) -- (6,-3.5) node[midway, right, scale = .66]{$\ \ 21$};

\draw [->, thick] (B4) -- ( .5 ,-2.5) node[midway,above, scale = .66]{$10$};
\draw [->, thick] (B4) -- (6,-2.5) node[midway, right, scale = .66]{$\ \ 14$};
\draw [->, thick] (B4) -- (.5,0) node[midway, above, scale = .66]{$\ \ 35$};

 \end{tikzpicture}\caption{{\sc Drop-Out Bid Examples.} {\em In these two examples the dropout 
 bid vectors $(\beta_1,\beta_2,\beta_3, \beta_4)$ are 
 $(18,31, 31^+, 23)$ and $(24, 24, 24, 24^+)$, respectively.} }\label{fig:4-buyer-example}
\end{figure}
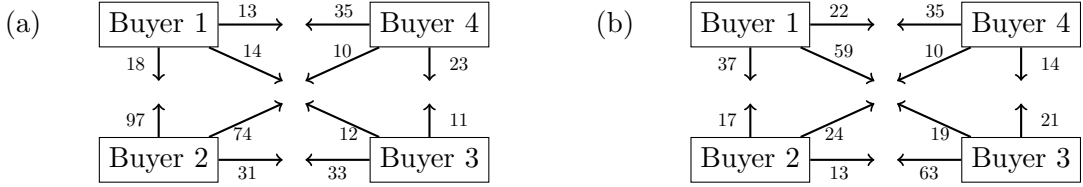

\subsubsection{Dropout Bids and the Iterative Deletion of Weakly Dominated Strategies}\label{sec:iterative-deletion}
As well as being solutions to the ascending price auction, the dropout bids have a much stronger property that
makes them the natural and robust prediction for auctions with interdependent valuations. Specifically, Paes Leme et al.~\cite{PST12}
proved that, for each buyer~$i$, the interval $[0,\beta_i]$ is the set of strategies that survive {\bf any sequence} consisting of the
iterative deletion of weakly dominated strategies. This is formalized as follows.
Take an $n$-buyer game with strategy sets $S_1, S_2,\dots, S_n$ and utility functions 
$u_i:S_1\times S_2\times \cdots \times S_n\rightarrow \mathbb{R}$.
Then $\{S_i^\tau\}_{i,\tau}$ is a {\em valid sequence} for the iterative deletion of weakly 
dominated strategies if for each $\tau$ there is a buyer~$i$
such that $(i) S_j^\tau=S_j^{\tau-1}$ for each buyer~$j\neq i$ and
$(ii) S_i^\tau\subset S_i^{\tau-1}$ where for each strategy $s_i\in S_i^{\tau-1}\setminus S_i^{\tau}$ there is an $\hat{s}_i\in S_i^{\tau}$
such that $u_i(\hat{s}_i, s_{-i})\ge u_i(s_i, s_{-i})$ for all $s_{-i}\in \prod_{j:j\neq i}\, S^{\tau}_j$, and with strict inequality for at least one
$s_{-i}$. 

We say that a strategy $s_i$ for buyer~$i$ {\em survives} the iterative deletion of weakly dominated strategies if for any valid sequence
 $\{S_i^\tau\}_{i,\tau}$ we have $s_i\in \bigcap_{\tau} S^\tau_i$.

\begin{theorem}\cite{PST12}\label{thm:first-multi}
Given a first-price auction with interdependent valuations, for each buyer~$i$, the set of bids that survive 
the iterative deletion of weakly dominated strategies is exactly $[0,\beta_i]$. \qed
\end{theorem}
An exact analogue of Theorem~\ref{thm:first-multi} does {\em not} hold for second-price auctions
with interdependent valuations. We prove this in Appendix~A where we present 
an example in which the set of strategies
that survive the iterative deletion of weakly dominated strategies is empty in a second-price auction.
However, consideration of that example shows that the problem occurs when a strategy is deleted because it
is weakly dominated by a {\bf higher} value bid. Observe that this can never happen for a potentially
winning bid in a first-price auction. Thus Theorem~\ref{thm:first-multi} still holds in first-price auctions
when we restrict attention to sequences consisting of the iterative deletion of strategies that are 
weakly dominated by a lower bid. Indeed, we can prove the corresponding theorem also holds
for second-price auctions.

\begin{theorem}\label{thm:second-multi}
Given a second-price auction with interdependent valuations, for each losing buyer~$i$, the set of bids that survive 
the iterative deletion of strategies that are weakly dominated by a lower bid is exactly $[0,\beta_i]$. \qed
\end{theorem}
\begin{proof}
First we claim that for any losing buyer~$i$ and any price $p>\beta_i$ there is a sequence of
iterative deletions of strategies that are weakly dominated by a lower bid that leads to the deletion of bid $p$ from $S^\tau_i$.
Without loss of generality,  we may order the buyers such that $\beta_1\le \beta_2 \le \cdots\le \beta_n$;
in the case of a tie the buyers are placed in the order they were deleted by the tie-breaking rule.
Initially $S_i^0=[0,\infty)$, for each buyer~$i$. 
We now define a valid sequence such that $S^i_i=[0,\beta_i]$.
We proceed by induction on the label of the buyers. For the base case observe that for Buyer~$1$ we know 
$\beta_1 = \max\limits_{j:j\neq i}\, (v_{i,i} - v_{i,j} )$ is the highest price it wants to pay to beat anyone else.
Suppose Buyer~$1$ bids $p>\beta_1$. Take any set of bids $b_{-1}\in \times_{j:j \ge 2} S^0_j$.
We have three cases: \\
(i) Both bids $p$ and $\beta_1$ are winning bids against  $b_{-1}$. Then, as this is a second-price auction,
Buyer~$1$ is indifferent between the two bids.\\
(ii) Both bids $p$ and $\beta_1$ are losing bids against  $b_{-1}$. Then 
Buyer~$1$ is indifferent between the two bids.\\
(iii) Bid $p$ is a winning bid but $\beta_i$ is a losing bid against $b_{-1}$. Then since the winning price
is at least $\beta_1$, Buyer~$1$ strictly prefers to lose rather than win.
Moreover, since $S_j^0=[0,\infty)$, there is a set of bids $b_{-1}$ by the other buyers such that
Buyer~$1$ strictly prefers to lose rather than win.

Thus the bid $p$ is weakly dominated by the lower bid $\beta_1$. Since this applies to any $p>\beta_1$,
in Step~$1$ we may delete every bid for Buyer~$1$ above $\beta_1$. Therefore
$S^1_1=[0,\beta_1]$ and $S^1_j=[0,\infty]$ for each buyer $j\ge 2$.

For the induction hypothesis assume $S^{i-1}_j=[0,\beta_j]$, for all $j<i$ and $S^{i-1}_j=[0,\infty)$, for all $j\ge i$.
Now take a losing buyer~$i$ and any set of bids $b_{-i}\in \times_{j:j \neq i} S^{i-1}_j$.
Again, we have three cases: \\
(i) Both bids $p$ and $\beta_i$ are winning bids against  $b_{-i}$. Then, as this is a second-price auction,
buyer~$i$ is indifferent between the two bids.\\
(ii) Both bids $p$ and $\beta_1$ are losing bids against  $b_{-i}$. Then 
buyer~$i$ is indifferent between the two bids.\\
(iii) Bid $p$ is a winning bid but $\beta_i$ is a losing bid against $b_{-i}$. Then since $\beta_i$ is a losing bid under the tie-breaking rule, 
it must be the case that the winning bid is from a buyer~$j$ where $j> i$. But, by definition of $\beta_i$, there is no
buyer~$j$, with $j>i$, that buyer~$i$ wishes to beat at price $\beta_i$.

So buyer~$i$ prefers the bid $\beta_i$ to the bid $p$. Moreover, since any buyer $j:j>i$ has $S^{i-1}_j=[0,\infty)$, this preference 
is strict for some feasible choice of bids for the other buyers. 
Thus, for buyer~$i$, the bid $p$ is weakly dominated by the lower bid $\beta_i$, and this applies to
every $p>\beta_i$.
Ergo, in Step~$i$ we may delete every bid for buyer~$i$ above $\beta_i$. Therefore
$S^{i}_j=[0,\beta_i]$, for all $j<i+1$ and $S^{i-1}_j=[0,\infty)$, for all $j\ge i+1$. The claim then follows by induction.
So, for any losing buyer~$i$ we have that no bid greater than $\beta_i$ survives the iterative deletion of strategies that
are weakly dominated by a lower bid.

Observe that the above arguments also apply for the winning buyer, that is, buyer~$n$. Except, as there are no higher indexed
buyers, it is not the case that $\beta_n$ strictly dominates any bid $p>\beta_n$. Indeed, buyer~$n$ is indifferent
between all bids in the range $[\beta_n,\gamma_n]$, where $\gamma_n$ is the maximum value the buyer has
for beating any buyer~$j$ with dropout bid $\beta_j=\beta_n$. Observe, $\gamma_n$ does exist and is at least $\beta_n$
by definition of the ascending price mechanism. Thus, for the winning bidder no bid greater than $\gamma_i$ survives the 
iterative deletion of strategies that are weakly dominated by a lower bid.

Second, we claim for any buyer~$i$ and any price $q< \beta_i$ there is no sequence of
iterative deletions of strategies that are weakly dominated by a lower bid that leads to the deletion of bid $q$ from the 
feasible strategy space of buyer~$i$.
If not, consider the first time $\tau$ that some buyer~$i$ has a value $q\in [0,\beta_i]$ deleted from $S^\tau_i$.
We may assume that $q$ is deleted because it is was weakly dominated by a lower bid $p<q$.
Now, by assumption, $[0,\beta_j] \subseteq S^{\tau-1}_j$, for each buyer~$j$. Furthermore, by definition, there is some 
buyer~$k$, with $k>i$ that buyer~$i$ wishes to beat at any price below $\beta_i$. In particular, Buyer~$i$ wishes to beat
Buyer~$k$ at price $p$.
But since $k>i$ we have $\beta_k\ge \beta_i$. Recall that $[0,\beta_k] \subseteq S^{\tau-1}_k$. It immediately follows
that there is a set of feasible bids $b_k\in (p,q)$ and $b_j=0$, for all $j\notin \{i,k\}$ such that Buyer~$i$ strictly prefers
to win against these bids. Specifically, the bid $q$ is not weakly dominated by
the bid $p$, a contradiction. 
\end{proof}

It follows that the dropout bids form the {\em focal} subgame perfect equilibrium for both first-price and second-price auctions
with interdependent valuations.

We are now almost ready to be able to find equilibria in the sequential auction experiments we will conduct. This, in turn, will allow us to 
present a sequential auction with non-monotonic prices. Before doing so, one final factor remains to be discussed 
regarding the transition from equilibria in auctions with interdependent valuations to equilibria in sequential auctions.

\subsection{Equilibria in Sequential Auctions}\label{sec:equilibria-sequential}
\subsubsection{Tie-Breaking Rules (and Data Structures!)}\label{sec:tie-breaking-1}
As stated, the dropout bid of each buyer is uniquely defined. However, our description of the ascending auction 
may leave some flexibility in the choice of winner. Specifically, it may be the case that simultaneously
more than one buyer wishes to drop out of the auction. If this happens at the end of the ascending price procedure then
any of these buyers could be selected as the winner. An example of this is shown in Figure~\ref{fig:tie-breaking}.
 \begin{figure}[!h]
\centering
 \begin{tikzpicture}[scale=0.33]
 \node[rectangle,draw](B1) at (0,0) {Buyer 1};
\node[rectangle,draw, scale=1](B2) at (-6,-6) {Buyer 2};
\node[rectangle,draw, scale=1](B3) at (6,-6) {Buyer 3};

\draw [->,  thick] (B1) -- (-2.5,-2.5) node[midway, left, scale = .66]{$15\ \ $};
\draw [->, thick] (B2) -- ( -3.5 ,-3.5) node[midway, left, scale = .66]{$34\ $};
\draw [->, thick] (B1) -- (2.5, -2.5) node[midway,right, scale = .66]{$\ \ 15$};
\draw [->,  thick] (B2) -- (-.5,-6) node[midway, below, scale = .66]{$13$};
\draw [->, thick] (B3) -- ( .5 ,-6) node[midway,below, scale = .66]{$0$};
\draw [->, thick] (B3) -- (3.5,-3.5) node[midway,right, scale = .66]{$\ 126$};
 \end{tikzpicture}
\caption{{\sc Tie-Breaking.} {\em An example that requires tie-breaking to determine the winner. The drop-out bid 
vector is $(15, 15, 15)$ but there are two possible winners, that is, either $(15, 15^+, 15)$  or $(15, 15, 15^+)$.}}\label{fig:tie-breaking}
\end{figure}
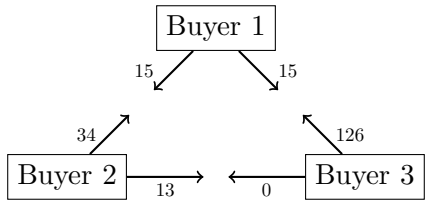

This observation implies that to fully define the ascending auction procedure we must incorporate a tie-breaking rule
to order the buyers when more than one wish to drop out simultaneously. In an auction with interdependent valuations
the tie-breaking rule only affects the choice of winner, but otherwise has no structural significance. 
However, in a sequential auction the choice of tie-breaking rule may have much more significant consequences.
Specifically, because each node in the game tree corresponds to an auction with interdependent valuations, the choice of
winner at one node may effect the valuations at nodes higher in the tree. In particular, the equilibrium path may vary 
with different tie-breaking rules, leading to different prices, winners, and utilities.

As we will show in Section~\ref{sec:tie-breaking-2} there are a massive number of tie-breaking rules, even in
small sequential auctions. 
We emphasize, however, that our main result holds regardless of the
tie-breaking rule. That is, for {\bf any} tie-breaking rule there is a sequential auction
on which it induces a non-monotonic price trajectory. This we will also show in Section~\ref{sec:general} after explaining 
mathematically how to classify every tie-breaking rule in terms of labelled, directed acyclic graphs.
First, though, we will show that non-monotonic pricing occurs on the equilibrium path for perhaps the
three most natural choices of tie-breaking rule, namely {\tt preferential-ordering}, {\tt first-in-first-out} and
{\tt last-in-first-out}. Interestingly these rules correspond to the fundamental data structures 
of priority queues, queues, and stacks used in computer science.

\subsubsection{Tie-Breaking Rules and Data Structures}\label{sec:data-structures}

\noindent\textbf{Preferential Ordering (Priority Queue): }In {\tt preferential-ordering} each buyer is given a distinct rank. In the case of a tie the buyer with the worst
rank is eliminated. Without loss of generality, we may assume that the ranks corresponding to 
a lexicographic ordering of the buyers. That is, the rank of a buyer is its index label and given a tie
amongst all the buyers that wish to dropout of the auction we remove the buyer with the highest index.
The preferential ordering tie-breaking rule corresponds to the data structure known as a {\em priority queue}. \\

\noindent\textbf{First-In-First-Out (Queue): }The {\tt  first-in-first-out} tie-breaking rule corresponds to the data structure known as a {\em queue}. 
The queue consists of those buyers in the auction that wish to dropout. Amongst these, the
buyer at the front of the queue is removed. If multiple buyers request to be added to the queue
simultaneously, they will be added lexicographically. Note though that this is different from preferential ordering
as the entire queue will not, in general, be ordered lexicographically. For example, when at a fixed price $p$ we 
remove the buyer~$i$ at the front of the queue this may cause new buyers to wish to dropout at price $p$ (i.e. those
buyers who only wanted to beat buyer~$i$). These new buyers will be placed behind the other buyers already in the queue. \\

\noindent\textbf{Last-In-First-Out (Stack): }The {\tt last-in-first-out} tie-breaking rule corresponds to the data structure known as a {\em stack}. 
Again the stack consists of those buyers in the auction that wish to dropout. Amongst these, the
buyer at the top of the stack (i.e. the back of the queue) is removed. If multiple buyers request to be added to the 
stack simultaneously, they will be added lexicographically. At first glance, this {\tt last-in-first-out} rule 
appears more unusual than the previous two, but it still has a natural interpretation in terms of an auction.
Namely, it corresponds to settings where the buyer whose situation has changed most recently reacts the 
quickest. \\

In order to understand these tie-breaking rules it is useful to see how they apply on an example.
In Figure~\ref{fig:example-tie-breaking} the dropout vector is 
$(\beta_1,\beta_2,\beta_3, \beta_4, \beta_5) =(40,40,40,40,40)$, but the three
tie-breaking rules will select three different winners.

 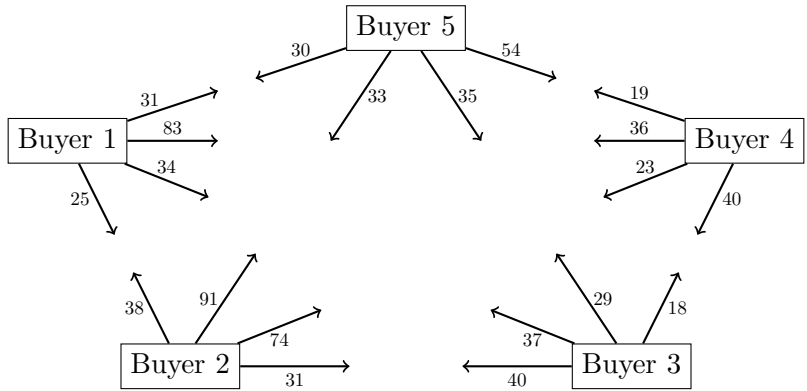
\begin{figure}[!h]
\centering
 \begin{tikzpicture}[scale=0.5]
 \node[rectangle,draw](B1) at (-9,0) {Buyer 1};
\node[rectangle,draw, scale=1](B2) at (-6,-6) {Buyer 2};
\node[rectangle,draw, scale=1](B3) at (6,-6) {Buyer 3};
\node[rectangle,draw, scale=1](B4) at (9,0) {Buyer 4};
 \node[rectangle,draw](B5) at (0,3) {Buyer 5};

\draw [->,  thick] (B1) -- (-7.75, -2.5) node[midway, left, scale = .66]{$25$};
\draw [->,  thick] (B1) -- (-5,0) node[midway, above, scale = .66]{$83$};
\draw [->, thick] (B1) -- (-5.25, -1.5) node[midway,above, scale = .66]{$34$};
\draw [->, thick] (B1) -- (-5, 1.33) node[near start,above, scale = .66]{$31$};

\draw [->, thick] (B2) -- ( -7.25 ,-3.5) node[midway,left, scale = .66]{$38$};
\draw [->,  thick] (B2) -- (-1.5,-6) node[midway, below, scale = .66]{$31$};
\draw [->,  thick] (B2) -- (-2.25,-4.5) node[midway, below, scale = .66]{$74$};
\draw [->, thick] (B2) -- (-4, -3) node[midway,left, scale = .66]{$91$};

\draw [->, thick] (B3) -- (1.5 ,-6) node[midway,below, scale = .66]{$40$};
\draw [->, thick] (B3) -- (2.25,-4.5) node[midway, below, scale = .66]{$37$};
\draw [->, thick] (B3) -- (7.25,-3.5) node[midway, right, scale = .66]{$18$};
\draw [->, thick] (B3) -- (4, -3) node[midway,right, scale = .66]{$29$};

\draw [->, thick] (B4) -- (7.75,-2.5) node[midway, right, scale = .66]{$40$};
\draw [->, thick] (B4) -- (5,0) node[midway, above, scale = .66]{$36$};
\draw [->, thick] (B4) -- (5.25,-1.5) node[midway,above, scale = .66]{$23$};
\draw [->, thick] (B4) -- (5, 1.33) node[midway,above, scale = .66]{$19$};

\draw [->, thick] (B5) -- ( -4 ,1.66) node[midway,above, scale = .66]{$30$};
\draw [->, thick] (B5) -- (-2,0) node[midway, right, scale = .66]{$33$};
\draw [->, thick] (B5) -- (2,0) node[midway, right, scale = .66]{$35$};
\draw [->, thick] (B5) -- (4, 1.66) node[midway,above, scale = .66]{$54$};

 \end{tikzpicture}
 \caption{An Example to Illustrate the Three Tie-Breaking Rules.} \label{fig:example-tie-breaking}
\end{figure}

On running the ascending price procedure, both Buyer~$3$ and Buyer~$4$ wish to drop out when the price reaches $40$. 
In {\tt preferential-ordering}, our choice set is then $\{3,4\}$ and we remove the highest index buyer, namely Buyer~$4$.
With the removal of Buyer~$4$, neither Buyer~$1$ nor Buyer~$5$ have an incentive to continue bidding so
they both decide to dropout. Thus our choice set is now $\{1, 3,5\}$ and {\tt preferential-ordering} removes Buyer~$5$.
Observe, with the removal of Buyer~$5$, that Buyer~$2$ no longer has an active participant it wishes to beat so the 
choice set is updated to $\{1, 2, 3\}$. The {\tt preferential-ordering} rule now removes the buyers in the order
Buyer~$3$, then Buyer~$2$ and lastly Buyer~$1$. Thus Buyer~$1$ wins under the {\tt preferential-ordering} rule.

Now consider {\tt first-in-first-out}. To allow for a consistent comparison between the three methods, we assume that when 
multiple buyers are simultaneously
added to the queue they are added in decreasing lexicographical order. Thus our initial queue is $4:3$ and {\tt first-in-first-out}
removes Buyer~$4$ from the front of the queue. With the removal of Buyer~$4$, neither Buyer~$1$ nor Buyer~$5$ have an 
incentive to continue bidding so they are added to the back of the queue. Thus the queue is now 
$3:5:1$ and {\tt first-in-first-out} removes Buyer~$3$ from the front of the queue. 
It then removes Buyer~$5$ from the front of the queue. With the removal of Buyer~$5$, we again have that Buyer~$2$ now wishes
to dropout. Hence the queue is  
$1:2$ and {\tt first-in-first-out} then removes Buyer~$1$ from the front of the queue and lastly removes Buyer~$2$.
Thus Buyer~$2$ wins under the  {\tt first-in-first-out} rule.

Finally, consider the {\tt last-in-first-out} rule. Again, to allow for a consistent comparison we assume that when multiple buyers are simultaneously
added to the stack they are added in increasing lexicographical order. Thus our initial stack is 
$\begin{smallmatrix} 4\\ 3 \end{smallmatrix}$ and {\tt last-in-first-out}
removes Buyer~$4$ from the top of the stack. Again, Buyer~$1$ and Buyer~$5$ both now wish to drop out so our stack becomes
$\begin{smallmatrix} 5\\1\\ 3 \end{smallmatrix}$. Therefore Buyer~$5$ is next removed from the the top of the stack.
At this point, Buyer~$2$ wishes to dropout so the stack becomes $\begin{smallmatrix} 2\\1\\ 3 \end{smallmatrix}$.
The {\tt last-in-first-out} rule now removes the buyers in the order
Buyer~$2$, then Buyer~$1$ and lastly Buyer~$3$. Thus Buyer~$3$ wins under the  {\tt last-in-first-out} rule.

We have now developed all the tools required to implement our sequential auction experiments. We describe these experiments and their results in
Section~\ref{sec:expts}. Before doing so, we present in Section~\ref{sec:counter-example} one sequential auction obtained via these experiments and verify that it leads 
to a non-monotonic price trajectory with each of the three tie-breaking rules discussed above. We then explain in Section~\ref{sec:general} how to generalize this
conclusion to apply to every tie-breaking rule. 

\section{An Auction with Non-Monotonic Prices}\label{sec:counter-example}
Here we prove that the decreasing price anomaly is {\bf not} guaranteed for sequential auctions with more than two buyers.
Specifically, in Section~\ref{sec:general} we prove the following result:
\begin{theorem}\label{thm:general}
For any tie-breaking rule, there is a sequential auction with non-monotonic prices.
\end{theorem}
In the rest of this section, we show that for all three of the tie-breaking rules discussed (namely, {\tt preferential-ordering}, {\tt first-in-first-out} and 
{\tt last-in-first-out}) there is a sequential auction with with non-monotonic prices.
Specifically, we exhibit a sequential auction with three buyers and eight items that exhibits non-monotonic prices. 

\begin{theorem}\label{thm:non-monotonic}
There is a sequential auction with non-monotonic prices for the {\tt preferential-ordering}, {\tt first-in-first-out} and 
{\tt last-in-first-out} tie-breaking rules. 
\end{theorem}
\begin{proof}
Our counter-example to the conjecture is a sequential auction with three buyers and 
eight identical items for sale. We present the first-price version where at equilibrium the buyers bid their 
dropout values in each time period; as discussed, the same example extends to second-price auctions.

The valuations of the three buyers are defined as follows.
Buyer 1 has marginal valuations $\{55,$ $55,$ $55,$ $55,$ $30,$ $20,$ $0,$ $0\}$, 
Buyer 2 has marginal valuations $\{32,$ $20,$ $0,$ $0,$ $0,$ $0,$ $0,$ $0\}$, and 
Buyer 3 has marginal valuations $\{44,$ $44,$ $44,$ $44,$ $0,$ $0,$ $0,$ $0\}$.

Let's now compute the extensive forms of the auction under the three tie-breaking rules.
We begin with the {\tt preferential-ordering} rule. To compute its extensive form, 
observe that Buyer~1 is guaranteed to win at least two items in the auction because Buyer~2 and Buyer~3
together have positive value for six items. Therefore, the feasible set of sink nodes 
in the extensive form representation are shown in Figure~\ref{fig:sink}. 

 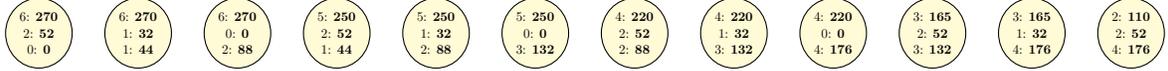
\begin{figure}[!h]
\centering
 \begin{tikzpicture}[scale=0.33]
\node[ellipse,draw, fill=yellow!20, align=center, scale=0.45](H1) at (1,0) { 6: {\bf 270} \\ 2: {\bf 52} \\ 0: {\bf 0}};
\node[ellipse,draw, fill=yellow!20, align=center, scale=0.45](H2) at (5,0) { 6: {\bf 270} \\ 1: {\bf 32} \\ 1: {\bf 44}};
\node[ellipse,draw, fill=yellow!20,align=center, scale=0.45](H3) at (9,0) { 6: {\bf 270} \\ 0: {\bf 0} \\ 2: {\bf 88}};
\node[ellipse,draw, fill=yellow!20,align=center, scale=0.45](H4) at (13,0) { 5: {\bf 250} \\ 2: {\bf 52} \\ 1: {\bf 44}};
\node[ellipse,draw, fill=yellow!20,align=center, scale=0.45](H5) at (17,0) { 5: {\bf 250} \\ 1: {\bf 32} \\ 2: {\bf 88}};
\node[ellipse,draw,fill=yellow!20, align=center, scale=0.45](H6) at (21,0) { 5: {\bf 250} \\ 0: {\bf 0} \\ 3: {\bf 132}};
\node[ellipse,draw, fill=yellow!20,align=center, scale=0.45](H7) at (25,0) { 4: {\bf 220} \\ 2: {\bf 52} \\ 2: {\bf 88}};
\node[ellipse,draw, fill=yellow!20,align=center, scale=0.45](H8) at (29,0) { 4: {\bf 220} \\ 1: {\bf 32} \\ 3: {\bf 132}};
\node[ellipse,draw, fill=yellow!20,align=center, scale=0.45](H9) at (33,0) { 4: {\bf 220} \\ 0: {\bf 0} \\ 4: {\bf 176}};
\node[ellipse,draw, fill=yellow!20,align=center, scale=0.45](H10) at (37,0) { 3: {\bf 165} \\ 2: {\bf 52} \\ 3: {\bf 132}};
\node[ellipse,draw, fill=yellow!20,align=center, scale=0.45](H11) at (41,0) { 3: {\bf 165} \\ 1: {\bf 32} \\ 4: {\bf 176}};
\node[ellipse,draw, fill=yellow!20,align=center, scale=0.45](H12) at (45,0) { 2: {\bf 110} \\ 2: {\bf 52} \\ 4: {\bf 176}};
\end{tikzpicture}
\caption{Sink Nodes of the Extensive Form Game.}\label{fig:sink}
\end{figure}

Given the valuations at the sink nodes we can work our way upwards recursively calculating the values at the 
other nodes in the extensive form representation. For example, consider the node $(x_1,x_2,x_3)=(4,1,2)$.
This node has three children, namely $(5,1,2), (4,2,2)$ and $(4,1,3)$; see Figure~\ref{fig:interdependent}(a).
These induce a three-buyer auction as shown in Figure~\ref{fig:interdependent}(b). This can be solved using the ascending price procedure
to find the dropout bids for each buyer. Thus we obtain that the value for the node $(x_1,x_2,x_3)=(4,1,2)$ is as shown in Figure~\ref{fig:interdependent}(c).
Of course this node is particularly simple as, for the final round of the sequential auction, the corresponding auction 
with interdependent valuations is just a standard auction. That is, when the final item is sold, for any buyer~$i$ the value $v_{i,j}$ is 
independent of the buyer $j\neq i$. 

 \begin{figure}[!h]
\centering
 \begin{tikzpicture}[scale=0.25]
  \node (a) at (-7,0) {(a)};
\node[ellipse,draw, align=center, scale=0.45](B2) at (0,0) { 4: {\bf \ \ -\ \ } \\ 1: {\bf \ \ -\ \ } \\ 2: {\bf \ \ -\ \ }};
\node[ellipse,draw, fill=yellow!20,align=center, scale=0.45](C2) at (-6,-6) { 5: {\bf 250} \\ 1: {\bf 32} \\ 2: {\bf 88}};
\node[ellipse,draw, fill=yellow!20,align=center, scale=0.45](C4) at (0,-6) { 4: {\bf 220} \\ 2: {\bf 52} \\ 2: {\bf 88}};
\node[ellipse,draw, fill=yellow!20,align=center, scale=0.45](C5) at (6,-6) { 4: {\bf 220} \\ 1: {\bf 32} \\ 3: {\bf 132}};

\draw [->] (B2) -- (C2) node[midway,left, scale = .66]{};
\draw [->] (B2) -- (C4) node[midway,left, scale = .66]{};
\draw [->] (B2) -- (C5) node[midway,right, scale = .66]{};
\end{tikzpicture}
\quad \quad 
 \begin{tikzpicture}[scale=0.25]
   \node (b) at (-8,0) {(b)};
 \node[rectangle,draw](B1) at (0,0) {Buyer 1};
\node[rectangle,draw, scale=1](B2) at (-6,-6) {Buyer 2};
\node[rectangle,draw, scale=1](B3) at (6,-6) {Buyer 3};

\draw [->,  thick] (B1) -- (-2.5,-2.5) node[midway, left, scale = .66]{$30\ $};
\draw [->, thick] (B2) -- ( -3.5 ,-3.5) node[left, scale = .66]{$20\ $};
\draw [->, thick] (B1) -- (2.5, -2.5) node[midway,right, scale = .66]{$\ \ 30$};
\draw [->,  thick] (B2) -- (-.5,-6) node[midway, below, scale = .66]{$20\ $};
\draw [->, thick] (B3) -- ( .5 ,-6) node[midway,below, scale = .66]{$44$};
\draw [->, thick] (B3) -- (3.5,-3.5) node[right, scale = .66]{$\ \ 44$};
 \end{tikzpicture}
 \quad \quad 
  \begin{tikzpicture}[scale=0.25]
  \node (c) at (-7,0) {(c)};
\node[ellipse,draw, align=center, scale=0.45](B2) at (0,0) { 4: {\bf 127} \\ 1: {\bf 32} \\ 2: {\bf 48}};

\node[ellipse,draw, fill=yellow!20,align=center, scale=0.45](C2) at (-6,-6) { 5: {\bf 250} \\ 1: {\bf 32} \\ 2: {\bf 88}};
\node[ellipse,draw, fill=yellow!20,align=center, scale=0.45](C4) at (0,-6) { 4: {\bf 220} \\ 2: {\bf 52} \\ 2: {\bf 88}};
\node[ellipse,draw, fill=yellow!20,align=center, scale=0.45](C5) at (6,-6) { 4: {\bf 220} \\ 1: {\bf 32} \\ 3: {\bf 132}};

\draw [->, dotted] (B2) -- (C2) node[midway,left, scale = .66]{$30\ $};
\draw [->, dotted] (B2) -- (C4) node[midway,left, scale = .66]{$20$};
\draw [->, thick] (B2) -- (C5) node[midway,right, scale = .66]{$30+$};
\end{tikzpicture}
\caption{Solving a Subgame above the Sinks.}\label{fig:interdependent}
\end{figure}
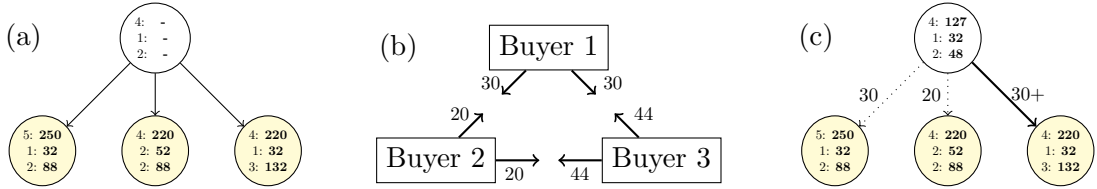

Nodes higher up the game tree correspond to more complex auctions with interdependent valuations.
For example, the case of the source node $(x_1,x_2,x_3)=(0,0,0)$ is shown in Figure~\ref{fig:interdependent2}.
In this case, on applying the ascending price procedure, Buyer~1 is the first to dropout at price $15$. At this point, both Buyer~2 and Buyer~3 no longer
have a competitor that they wish to beat at this price, so they both want to dropout. With the {\tt preferential-ordering} tie-breaking rule,
Buyer~$2$ wins the item.

 \begin{figure}[!h]
\centering
 \begin{tikzpicture}[scale=0.25]

\node[ellipse,draw, align=center, scale=0.45](R1) at (0,0) { 0: {\bf \ \ -\ \ } \\ 0: {\bf \ \ -\ \ } \\ 0: {\bf \ \ -\ \ }};

\node[ellipse,draw, align=center, scale=0.45](A1) at (-6,-6) { 1: {\bf 125} \\ 0: {\bf 1} \\ 0: {\bf 66}};
\node[ellipse,draw, align=center, scale=0.45](A2) at (0,-6) { 0: {\bf 110} \\ 1: {\bf 35} \\ 0: {\bf 176}};
\node[ellipse,draw, align=center, scale=0.45](A3) at (6,-6) { 0: {\bf 110} \\ 0: {\bf 22} \\ 1: {\bf 176}};

\draw [->, dotted] (R1) -- (A1) node[midway,left, scale = .66]{};
\draw [->, thick] (R1) -- (A2) node[midway,left, scale = .66]{};
\draw [->, dotted] (R1) -- (A3) node[midway,right, scale = .66]{};
\end{tikzpicture}
\quad \quad 
 \begin{tikzpicture}[scale=0.25]
 \node[rectangle,draw](B1) at (0,0) {Buyer 1};

\node[rectangle,draw, scale=1](B2) at (-6,-6) {Buyer 2};
\node[rectangle,draw, scale=1](B3) at (6,-6) {Buyer 3};

\draw [->,  thick] (B1) -- (-2.5,-2.5) node[midway, left, scale = .66]{$15\ $};
\draw [->, thick] (B2) -- ( -3.5 ,-3.5) node[left, scale = .66]{$34\ $};
\draw [->, thick] (B1) -- (2.5, -2.5) node[midway,right, scale = .66]{$\ \ 15$};
\draw [->,  thick] (B2) -- (-.5,-6) node[midway, below, scale = .66]{$13\ $};
\draw [->, thick] (B3) -- ( .5 ,-6) node[midway,below, scale = .66]{$0$};
\draw [->, thick] (B3) -- (3.5,-3.5) node[right, scale = .66]{$\ \ 110$};
 \end{tikzpicture}
 \quad \quad 
 \begin{tikzpicture}[scale=0.25]
\node[ellipse,draw, align=center, scale=0.45](R1) at (0,0) { 0: {\bf 110} \\ 0: {\bf 22} \\ 0: {\bf 176}};

\node[ellipse,draw, align=center, scale=0.45](A1) at (-6,-6) { 1: {\bf 125} \\ 0: {\bf 1} \\ 0: {\bf 66}};
\node[ellipse,draw, align=center, scale=0.45](A2) at (0,-6) { 0: {\bf 110} \\ 1: {\bf 35} \\ 0: {\bf 176}};
\node[ellipse,draw, align=center, scale=0.45](A3) at (6,-6) { 0: {\bf 110} \\ 0: {\bf 22} \\ 1: {\bf 176}};

\draw [->, dotted] (R1) -- (A1) node[midway,left, scale = .66]{$15\ $};
\draw [->, thick] (R1) -- (A2) node[midway,left, scale = .66]{$15+\ $};
\draw [->, dotted] (R1) -- (A3) node[midway,right, scale = .66]{$\ 15$};
\end{tikzpicture}
\caption{Solving the Subgame at the Root.}\label{fig:interdependent2}
\end{figure}
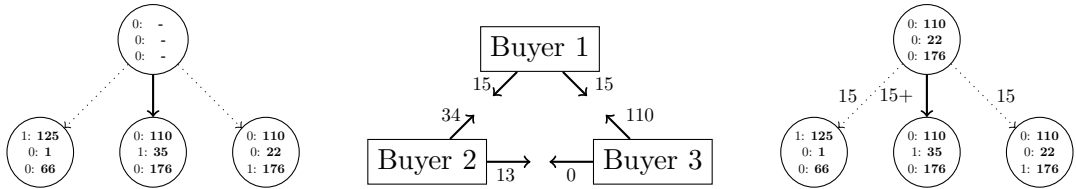

Using similar arguments at each node verifies the concise extensive form representation under the {\tt preferential-ordering} tie-breaking rule
shown Figure~\ref{fig:3-buyers}. In this figure, the white nodes represent subgames where the 
sequential auction still has three active buyers; the pink nodes represent subgames with at most two active 
buyers; the yellow nodes are the sink nodes. Again, the equilibrium path with non-monotonic prices
is shown in bold. Now consider this equilibrium path.
Observe that Buyer~$2$ wins the first two items, 
Buyer~$3$ wins the next four items and Buyer $1$ wins the final two items.
The resultant price trajectory is $\{15, 17, 0,0,0,0,0,0\}$.
That is, the price rises and then falls to zero -- a non-monotonic price trajectory.

Exactly the same example works with the other two tie-breaking rules. The extensive form representation 
with the {\tt first-in-first-out} rule is shown in Figure~\ref{fig:first-in-rule} in Appendix~B; the extensive form representation 
with the {\tt last-in-first-out} rule is shown in Figure~\ref{fig:last-in-rule} in Appendix~B.

\begin{figure}[!h]
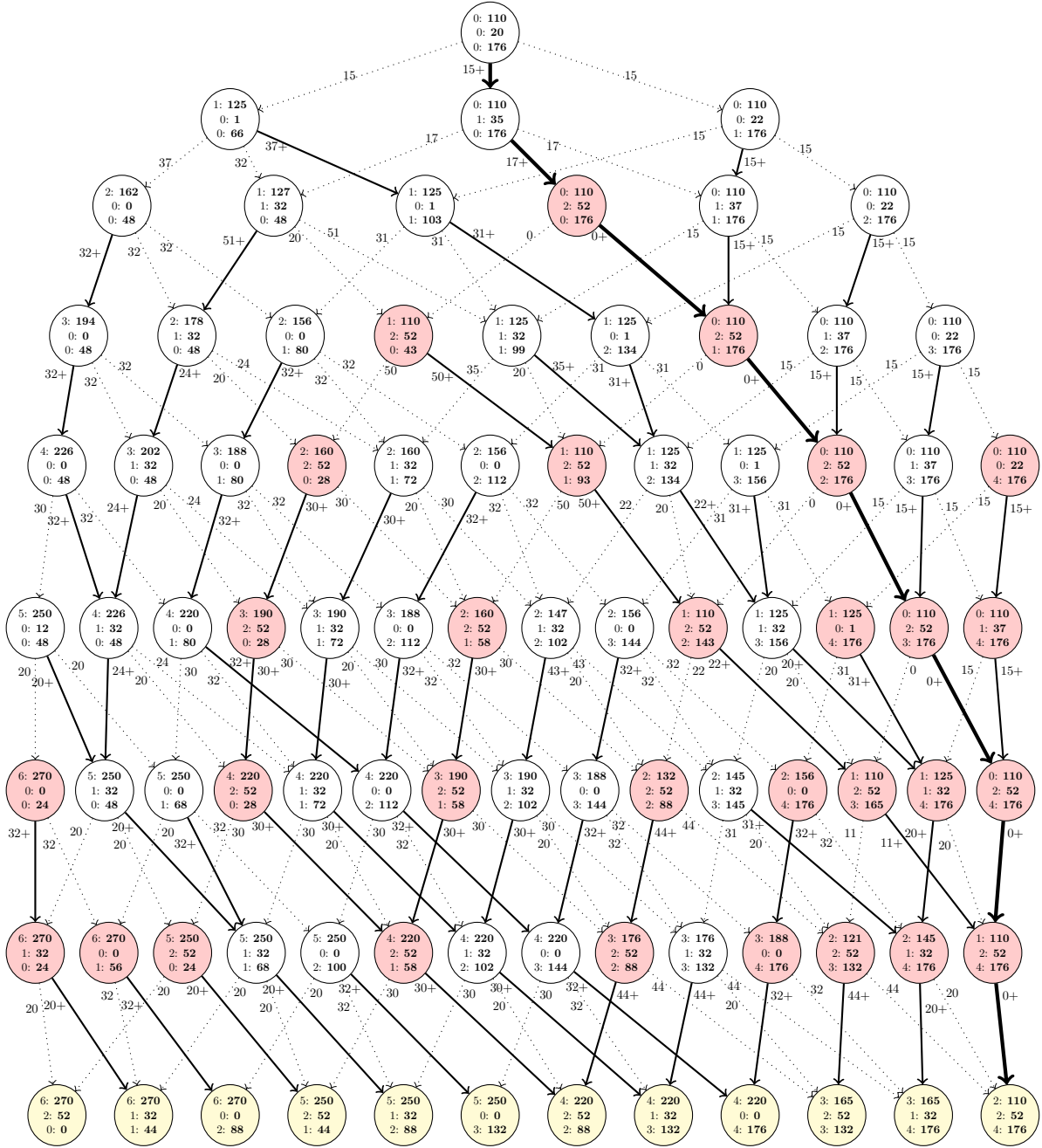

\centering

\caption{Non-Monotonic Prices with the {\tt preferential-ordering} Rule.}\label{fig:3-buyers}
\end{figure}

Notice the node values under {\tt preferential-ordering} and {\tt first-in-first-out} are exactly the same. 
This is despite the fact that these two rules do produce different winners at some nodes, for example the node $(3,0,2)$.
In contrast, the {\tt last-in-first-out} rule gives an extensive for where some nodes have different 
valuations than those produced by the other two rules. For example, for the node $(2,0,0)$ and its subgame the equilibrium 
paths and their prices differ in Figure~\ref{fig:first-in-rule} and Figure~\ref{fig:last-in-rule}. However, for all three rules the
equilibrium path and price trajectory {\em for the whole game} is exactly the same.
We remark that these observations will play a role when we prove that, for any tie-breaking rule, there
is a sequential auction with non-monotonic prices.
\end{proof}
\vspace{-6pt}

Again, we emphasize that there is nothing inherently perverse about this example.
The form of the valuation functions, namely decreasing marginal valuations, is standard.
As explained, the equilibrium concept studied is the appropriate one for sequential auctions.
Finally, the non-monotonic price trajectory is not the artifact of an aberrant tie-breaking rule;
we will now prove that non-monotonic prices are exhibited 
under any tie-breaking rule.

\section{Non-Monotonic Prices under General Tie-Breaking Rules}\label{sec:general}
Next we prove that for any tie-breaking rule there is a sequential auction on which it
produces a non-monotonic price trajectory. To do this, we must first formally define the set of
all tie-breaking rules.

\subsection{Classifying the set of Tie-Breaking Rules}\label{sec:tie-breaking-2}
Our definition of the set of tie-breaking rules will utilize the concept of an {\em overbidding graph}, introduced by \citet{PST12}.
For any price $p$ and any set of bidders $S$, the overbidding graph $G(S,p)$ contains a labelled vertex for each buyer in $S$
and an arc $(i,j)$ if and only if $v_{i,i}-p>v_{i,j}$. 
For example, recall the auction with interdependent valuations seen in Figure~\ref{fig:example-tie-breaking}.
This is reproduced in Figure~\ref{fig:DAG} along with its overbidding graph $G(\{1,2,3,4,5\}, 40)$ in Figure~\ref{fig:DAG}.

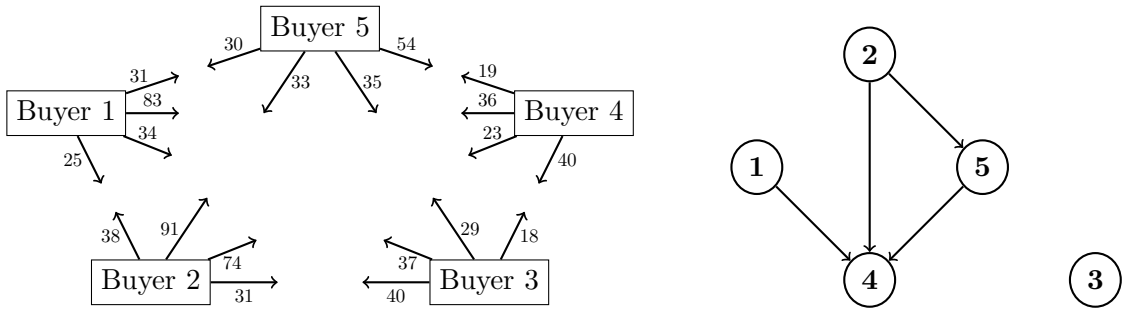
\begin{figure}[!h]
\centering
 \begin{tikzpicture}[scale=0.375]
\node[rectangle,draw](B1) at (-9,0) {Buyer 1};
\node[rectangle,draw, scale=1](B2) at (-6,-6) {Buyer 2};
\node[rectangle,draw, scale=1](B3) at (6,-6) {Buyer 3};
\node[rectangle,draw, scale=1](B4) at (9,0) {Buyer 4};
 \node[rectangle,draw](B5) at (0,3) {Buyer 5};

\draw [->,  thick] (B1) -- (-7.75, -2.5) node[midway, left, scale = .66]{$25$};
\draw [->,  thick] (B1) -- (-5,0) node[midway, above, scale = .66]{$83$};
\draw [->, thick] (B1) -- (-5.25, -1.5) node[midway,above, scale = .66]{$34$};
\draw [->, thick] (B1) -- (-5, 1.33) node[near start,above, scale = .66]{$31$};

\draw [->, thick] (B2) -- ( -7.25 ,-3.5) node[midway,left, scale = .66]{$38$};
\draw [->,  thick] (B2) -- (-1.5,-6) node[midway, below, scale = .66]{$31$};
\draw [->,  thick] (B2) -- (-2.25,-4.5) node[midway, below, scale = .66]{$74$};
\draw [->, thick] (B2) -- (-4, -3) node[midway,left, scale = .66]{$91$};

\draw [->, thick] (B3) -- (1.5 ,-6) node[midway,below, scale = .66]{$40$};
\draw [->, thick] (B3) -- (2.25,-4.5) node[midway, below, scale = .66]{$37$};
\draw [->, thick] (B3) -- (7.25,-3.5) node[midway, right, scale = .66]{$18$};
\draw [->, thick] (B3) -- (4, -3) node[midway,right, scale = .66]{$29$};

\draw [->, thick] (B4) -- (7.75,-2.5) node[midway, right, scale = .66]{$40$};
\draw [->, thick] (B4) -- (5,0) node[midway, above, scale = .66]{$36$};
\draw [->, thick] (B4) -- (5.25,-1.5) node[midway,above, scale = .66]{$23$};
\draw [->, thick] (B4) -- (5, 1.33) node[midway,above, scale = .66]{$19$};

\draw [->, thick] (B5) -- ( -4 ,1.66) node[midway,above, scale = .66]{$30$};
\draw [->, thick] (B5) -- (-2,0) node[midway, right, scale = .66]{$33$};
\draw [->, thick] (B5) -- (2,0) node[midway, right, scale = .66]{$35$};
\draw [->, thick] (B5) -- (4, 1.66) node[midway,above, scale = .66]{$54$};
\end{tikzpicture}
\quad \quad \quad
 \begin{tikzpicture}[scale=0.25]
\node[ellipse,thick,draw, align=center](R1) at (-6,6) {{\bf 1}};
\node[ellipse,thick,draw, align=center](R2) at (0,12) {{\bf 2}};
\node[ellipse,thick,draw, align=center](R3) at (12,0) {{\bf 3}};
\node[ellipse,thick,draw, align=center](R4) at (0,0) {{\bf 4}};
\node[ellipse,thick,draw, align=center](R5) at (6,6) {{\bf 5}};

\draw [->, thick] (R1) -- (R4);
\draw [->, thick] (R2) -- (R4);
\draw [->, thick] (R2) -- (R5);
\draw [->, thick] (R5) -- (R4);
\end{tikzpicture}
 \caption{The Overbidding Graph $G(\{1,2,3,4,5\}, 40)$.} \label{fig:DAG}
\end{figure}

But what does the overbidding graph have to do with tie-breaking rules?
First, recall that the drop-out bid $\beta_i$ is unique for any buyer~$i$, regardless of the 
tie-breaking rule. Consequently, whilst the tie-breaking rule will also be used to order buyers that are eliminated at prices 
below the final price $p^*$, such choices are irrelevant with regards to the final winner. Thus, the only relevant factor is
how a decision rule selects a winner from amongst those buyers $S^*$ whose drop-out bids are $p^*$. Second, recall
that at the final price $p^*$ the remaining buyers are eliminated one-by-one until there is a single winner. However, a buyer {\em cannot} 
be eliminated if there remains another buyer still in the auction that it wishes to beat at price $p^*$. That is, buyer~$i$ must be eliminated
after buyer~$j$ if there is an arc $(i,j)$ in the overbidding graph. Thus, the order of eliminations given by the tie-breaking rule
must be consistent with the overbidding graph. 
In particular, the winner can only be selected from amongst the {\em source vertices}\footnote{A {\em source} is a vertex $v$ with in-degree zero; that is, there
no arcs pointing into $v$.} in the overbidding graph $G(S^*, p^*)$. For example, in Figure~\ref{fig:DAG} the source vertices are $\{1, 2, 3\}$.
Note that this explains why the tie-breaking rules 
{\tt preferential-ordering}, {\tt first-in-first-out} and {\tt last-in-first-out} chose Buyer $1$, Buyer $2$ and Buyer $3$
as winners but none of them selected Buyer $4$ or Buyer $5$.
Observe that the overbidding graph $G(S^*, p^*)$ is {\em acyclic}; if it contained a directed cycle then the price in the ascending
auction would be forced to rise further. Because every directed acyclic graph contains at least one source vertex, any tie-breaking rule
does have at least one choice for winner.

Thus a tie-breaking rule is simply a function $\tau:H\rightarrow \sigma(H)$. 
Here the domain of the function is the set of labelled, directed acyclic graphs
and $\sigma(H)$ is the set of source nodes in $H$. Consequently, two tie-breaking rules
are equivalent if they correspond to the same function $\tau$.
We are now ready to present our main result.

\subsection{Non-Monotonic Prices for Any Tie-Breaking Rule}\label{sec:main}

\begin{theorem}
For any tie-breaking rule $\tau$, there is a sequential auction on which it produces non-monotonic prices.
\end{theorem}
\begin{proof}
We consider exactly the same example as in Theorem~\ref{thm:non-monotonic}.
That is, we have three buyers and eight items where
Buyer~$1$ has marginal valuations $\{55,55,55,55,30,20,0,0\}$, 
Buyer~$2$ has marginal valuations $\{32,20,0,0,0,0,0,0\}$, and 
Buyer~$3$ has marginal valuations $\{44,44,44,44,0,0,0,0\}$.

First let's calculate how many tie-breaking rules there are for this auction. To count this we must consider
all directed acyclic graphs with labels in $\{1,2,3\}$. Note that we must have at least two
buyers with drop-out values equal to the final price $p^*$ otherwise the auction would
have terminated earlier. Thus it suffices to consider directed acyclic graphs with either two or three vertices.
There are $8$ such topologies that produce $34$ labelled
directed acyclic graphs and $12,288$ tie-breaking rules! This is all illustrated in Table~\ref{tab:DAGs}.

\begin{table}[!h]
\centering
\begin{tabular}{|>{\centering}m{1.8cm}||c|c|c|c|c|c|c|c||>{\centering}m{1.6cm}|>{\centering}m{2cm}|}
 \hline
    Directed Acyclic Graph
 & 
 \begin{tikzpicture}[baseline=(current bounding box.west),scale=1]
    \node[circle,thick,draw](x) at (0,0) {x};
    \node[circle,thick,draw](y) at (0,-1.5) {y};
 \end{tikzpicture}
 &
 \begin{tikzpicture}[baseline=(current bounding box.west),scale=1]
    \node[circle,thick,draw](x) at (0,0) {x};
    \node[circle,thick,draw](y) at (0,-1.5) {y};
    \draw [->,  thick] (x) -- (y) ;
 \end{tikzpicture}
 &
 \begin{tikzpicture}[baseline=(current bounding box.west),scale=1]
    \node[circle,thick,draw](x) at (0,0) {x};
    \node[circle,thick,draw](y) at (0,-1) {y};
    \node[circle,thick,draw](z) at (0,-2) {z};
 \end{tikzpicture}
 &
 \begin{tikzpicture}[baseline=(current bounding box.west),scale=1]
    \node[circle,thick,draw](x) at (0,0) {x};
    \node[circle,thick,draw](y) at (0,-1) {y};
    \node[circle,thick,draw](z) at (0,-2) {z};
    \draw [->,  thick] (x) -- (y) ;
 \end{tikzpicture}
 &
 \begin{tikzpicture}[baseline=(current bounding box.west),scale=1]
    \node[circle,thick,draw](x) at (0,0) {x};
    \node[circle,thick,draw](y) at (0,-1) {y};
    \node[circle,thick,draw](z) at (0,-2) {z};
    \draw [->,  thick] (x) -- (y) ;
    \draw [->,  thick] (y) -- (z) ;
 \end{tikzpicture}
 &
 \begin{tikzpicture}[baseline=(current bounding box.west),scale=1]
    \node[circle,thick,draw](y) at (0,0) {y};
    \node[circle,thick,draw](z) at (0,-2) {z};
    \node[circle,thick,draw](x) at (0,-1) {x};
    \draw [->,  thick] (x) -- (y) ;
    \draw [->,  thick] (x) -- (z) ;
 \end{tikzpicture}
 &
 \begin{tikzpicture}[baseline=(current bounding box.west),scale=1]
    \node[circle,thick,draw](x) at (0,-1) {x};
    \node[circle,thick,draw](y) at (0,0) {y};
    \node[circle,thick,draw](z) at (0,-2) {z};
    \draw [->,  thick] (y) -- (x) ;
    \draw [->,  thick] (z) -- (x) ;
 \end{tikzpicture}
 &
 \begin{tikzpicture}[baseline=(current bounding box.west),scale=1]
    \node[circle,thick,draw](x) at (0.5,1) {x};
    \node[circle,thick,draw](y) at (0,0) {y};
    \node[circle,thick,draw](z) at (0.5,-1) {z};
    \draw [->,  thick] (x) -- (y) ;
    \draw [->,  thick] (y) -- (z) ;
    \draw [->,  thick] (x) -- (z) ;
 \end{tikzpicture}
 &
 Total $\#$ Labelled DAGs
 &
 Total $\#$ Tie-Breaking Rules
 \tabularnewline
 \hline
 $\#$ Labelled Graphs & 3 & 6 & 1 & 6 & 6 & 3 & 3 & 6 & 34 & \multirow{2}{2cm}{$1^{21}\cdot$ $2^{12}\cdot$ $3^1$ $= 12,288$} \tabularnewline
 \cline{1-10}
 $\#$ Sources & 2 & 1 & 3 & 2 & 1 & 1 & 2 & 1 & & \tabularnewline
 \hline
\end{tabular}
\caption{Labelled Directed Acyclic Graphs}\label{tab:DAGs}

\end{table}

Luckily we do not need to examine all these tie-breaking rules separately. It turns out 
that the set of tie-breaking rules can be partitioned into exactly {\bf ten} classes. Specifically, any tie-breaking rule produces 
one of just ten possible (in terms of distinct node valuations) extensive forms for this sequential auction. 
Two of these we have seen before.
The first is the extensive form shown in Figure~\ref{fig:3-buyers} (and also shown in Figure~\ref{fig:first-in-rule}) that is  
produced by both {\tt preferential-ordering} and {\tt first-in-first-out}.
The second is the extensive form shown in Figure~\ref{fig:last-in-rule} produced by {\tt last-in-first-out}.

Let's explain why there are only eight other feasible extensive forms. 
For any tie-breaking rule, as we work up from the sink nodes there are many nodes 
where the tie-breaking rule is required. Given this fact, why doesn't the total number of distinct extensive forms blow-up 
multiplicatively?
As previously alluded to, when we apply a tie-breaking rule there are two possibilities that arise.
In the first possibility, the node valuations are the same regardless of which buyer is selected by the rule.
Indeed this is why {\tt preferential-ordering} and {\tt first-in-first-out} can produce the same
extensive form. For an example, consider the node $(3,0,2)$ where Buyer~$1$ wins with 
{\tt preferential-ordering} but Buyer~$3$ wins with {\tt first-in-first-out}; in either case the
node valuations are identical, namely $(188,0,112)$ as shown in Figures~\ref{fig:3-buyers} and~\ref{fig:first-in-rule}.
For our purpose, such nodes are no importance.

\begin{figure}[h]
\centering
  \begin{tikzpicture}[scale=1]
   \node (A) at (0,0) {(A)};
   \node (B) at (4,0) {(B)};
   \node (C) at (8,0) {(C)};
   \node (D) at (12,0) {(D)};
    \node[circle,thick,draw](A1) at (2,-1) {1};
    \node[circle,thick,draw](A2) at (1,0) {2};
    \node[circle,thick,draw](A3) at (3,0) {3};
    \draw [->, thick] (A2) -- (A1);
    \draw [->, thick] (A3) -- (A1);
    \node[circle,thick,draw](B2) at (5.25,-0.5) {2};
    \node[circle,thick,draw](B3) at (6.75,-0.5) {3};
    \node[circle,thick,draw](C1) at (9,-0.5) {1};
    \node[circle,thick,draw](C2) at (10,-0.5) {2};    
    \node[circle,thick,draw](C3) at (11,-0.5) {3};
    \node[circle,thick,draw](D1) at (13,0) {1};
    \node[circle,thick,draw](D2) at (15,0) {2};
    \node[circle,thick,draw](D3) at (14,-1) {3};
    \draw [->, thick] (D1) -- ( D3) ;
 \end{tikzpicture}
\caption{The Four Critical Overbidding Graphs.}\label{fig:four}
\end{figure}
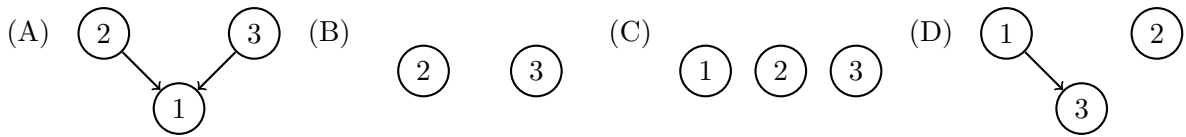

In the second possibility, the node valuations do vary depending upon which buyer is selected by the tie-breaking rule.
It turns out, however, that of the $34$ labelled directed acyclic graphs only $4$ of these overbidding graphs affect the extensive form
node valuations.  These four critical overbidding graphs, which we call $A,B,C$ and $D$, are shown in Figure~\ref{fig:four}.

So why are these these the only four overbidding graphs that matter?
The reader may verify that, working upwards from the sink nodes, the first such nodes where the choice of tie-breaking rule
matters occur at depth~$4$. Specifically, at the three nodes $(4,0,0), (1,0,3)$ and $(0,1,3)$. 
Now the nodes $(1,0,3)$ and $(0,1,3)$ both correspond to the 
overbidding graph $A$ whilst the node $(4,0,0)$ corresponds to the overbidding graph $B$.
For the overbidding graph $A$ the tie-breaking rule must select either the sink vertex~$2$ or the sink vertex~$3$ to win.
Moreover, by definition, it must make the same choice at both $(1,0,3)$ and $(0,1,3)$.
Furthermore, regardless of this choice, as we work up the extensive form the nodes 
$(1,0,2)$, $(0,1,2)$, $(0,0,3)$, $(0,1,1)$, $(0,0,2)$, $(0,1,0)$, $(0,0,1)$ and $(0,0,0)$
also all have the overbidding graph $A$ and, thus, must also have the same winner.

\tikzstyle{decision} = [diamond, draw, fill=blue!20, 
    text width=5em, text badly centered, node distance=3cm, inner sep=0pt]
\tikzstyle{block} = [rectangle, draw, fill=blue!20, 
    text width=5.5em, text centered, rounded corners, minimum height=3em]
\tikzstyle{line} = [draw, -latex']
\tikzstyle{cloudA} = [draw, ellipse,fill=red!20, node distance=3cm,
    minimum height=2em]
    \tikzstyle{cloudB} = [draw, ellipse,fill=green!20, node distance=3cm,
    minimum height=2em]
\tikzstyle{box} = [rectangle, draw, node distance=3cm, text width=5.5em]

    \begin{figure}[!h]
\centering
\begin{tikzpicture}[node distance = 2cm, auto]
    \node [block] (400) {Node (4,0,0) \\ {\tt DAG: B}};
    
    \node [block, xshift=-2.5cm, yshift=.5cm, left of=400, below of =400] (300a) {Node (3,0,0) {\tt DAG: D}};
    \node [block, xshift= 1.5cm, yshift=.5cm, right of=400, below of=400] (300b) {Node (3,0,0) {\tt DAG: C}};
  
    \node [block, yshift=.5cm, left of=300a, below of =300a] (200a) {Node (2,0,0) {\tt DAG: D}};
    \node [block, xshift=-.5cm, yshift=.5cm, right of=300a, below of =300a] (200b) {Node (2,0,0) {\tt DAG: D}};
    \node [block, xshift=-1cm,yshift=.5cm, left of =300b, below of =300b] (200c) {Node (2,0,0) {\tt DAG: C}};
    \node [block, xshift=-.5cm,yshift=.5cm, right of=300b, below of =300b] (200d) {Node (2,0,0) {\tt DAG: D}};         
  
    \node [block, below of =200a] (103a) {Node (1,0,3) \\ {\tt DAG: A}};
    \node [block, below of =200b] (103b) {Node (1,0,3) \\ {\tt DAG: A}};
    \node [block, below of =200c] (103c) {Node (1,0,3) \\ {\tt DAG: A}};
    \node [block, xshift= -1.5cm, below of =200d] (103d) {Node (1,0,3) \\ {\tt DAG: A}};
    \node [block, xshift= 1.5cm, below of =200d] (103e) {Node (1,0,3) \\ {\tt DAG: A}};
  
    \node [cloudA, xshift=-0.75cm, yshift=0.75cm, below of =103a] (D1) {{\tt No}};     
    \node [cloudB, xshift=0.75cm, yshift=0.75cm, below of =103a] (D2) {{\tt Yes}}; 
    \node [cloudA, xshift=-0.75cm, yshift=0.75cm, below of =103b] (D3) {{\tt No}};     
    \node [cloudB, xshift=0.75cm, yshift=0.75cm, below of =103b] (D4) {{\tt Yes}}; 
    \node [cloudA, xshift=-0.75cm, yshift=0.75cm, below of =103c] (D5) {{\tt No}};     
    \node [cloudB, xshift=0.75cm, yshift=0.75cm, below of =103c] (D6) {{\tt Yes}}; 
    \node [cloudA, xshift=-0.75cm, yshift=0.75cm, below of =103d] (D7) {{\tt No}};     
    \node [cloudB, xshift=0.75cm, yshift=0.75cm, below of =103d] (D8) {{\tt Yes}}; 
    \node [cloudA, xshift=-0.75cm, yshift=0.75cm, below of =103e] (D9) {{\tt No}};     
    \node [cloudB, xshift=0.75cm, yshift=0.75cm, below of =103e] (D10) {{\tt Yes}}; 
         
    \path [line] (400) -| node[anchor=east] {2 wins}(300a);
    \path [line] (400) -| node[anchor=west] {3 wins}(300b);
       
    \path [line] (300a) -| node[anchor=east] {1 wins}(200a);
    \path [line] (300a) -| node[anchor=west] {2 wins}(200b);
    \path [line] (300b) -| node[anchor=east,yshift=-.5cm] {1 or 3 wins}(200c);
    \path [line] (300b) -| node[anchor=west] {2 wins}(200d);
           
    \path [line] (200a) -- node[anchor=east] {1 wins}(103a);
    \path [line] (200b) --  node[anchor=east] {2 wins}(103b);
    \path [line] (200c) -- node[anchor=east] {1 or 3 wins}(103c);
    \path [line] (200d) --  node[anchor=east] {1 wins}(103d);
    \path [line] (200d) -- node[anchor=west] {2 wins}(103e);
                  
    \path [line] (103a) --  node[anchor=east] {2 wins}(D1);
    \path [line] (103a) -- node[anchor=west] {3 wins}(D2);
    \path [line] (103b) --  node[anchor=east, near start] {2 wins}(D3);
    \path [line] (103b) -- node[anchor=west, near start] {3 wins}(D4);
    \path [line] (103c) --  node[anchor=east] {2 wins}(D5);
    \path [line] (103c) -- node[anchor=west] {3 wins}(D6);
    \path [line] (103d) --  node[anchor=east, near start] {2 wins}(D7);
    \path [line] (103d) -- node[anchor=west, near start] {3 wins}(D8);
    \path [line] (103e) --  node[anchor=east] {2 wins}(D9);
    \path [line] (103e) -- node[anchor=west] {3 wins}(D10);
              
    \node [box, yshift=1cm, below of = D1] (PO) {{\footnotesize {\tt Preferential Ordering; First-in -first-out}}};                
    \node [box, yshift=1cm, below of = D3] (Li) {{\footnotesize {\tt Last-in -first-out}}};    
    \node [box, yshift=1cm, below of = D6] (Mono) {{\footnotesize {\tt See Figure~\ref{fig:monotonic}}}};     
                  
    \draw [dotted] (D1) -- (PO);  
    \draw [dotted] (D3) -- (Li);  
    \draw [dotted] (D6) -- (Mono);       
\end{tikzpicture}
\caption{Monotonic Prices: Yes or No? {\em A Decision Tree Partitioning the Tie-Breaking Rules into Ten Classes}.}\label{fig:decision}
\end{figure}
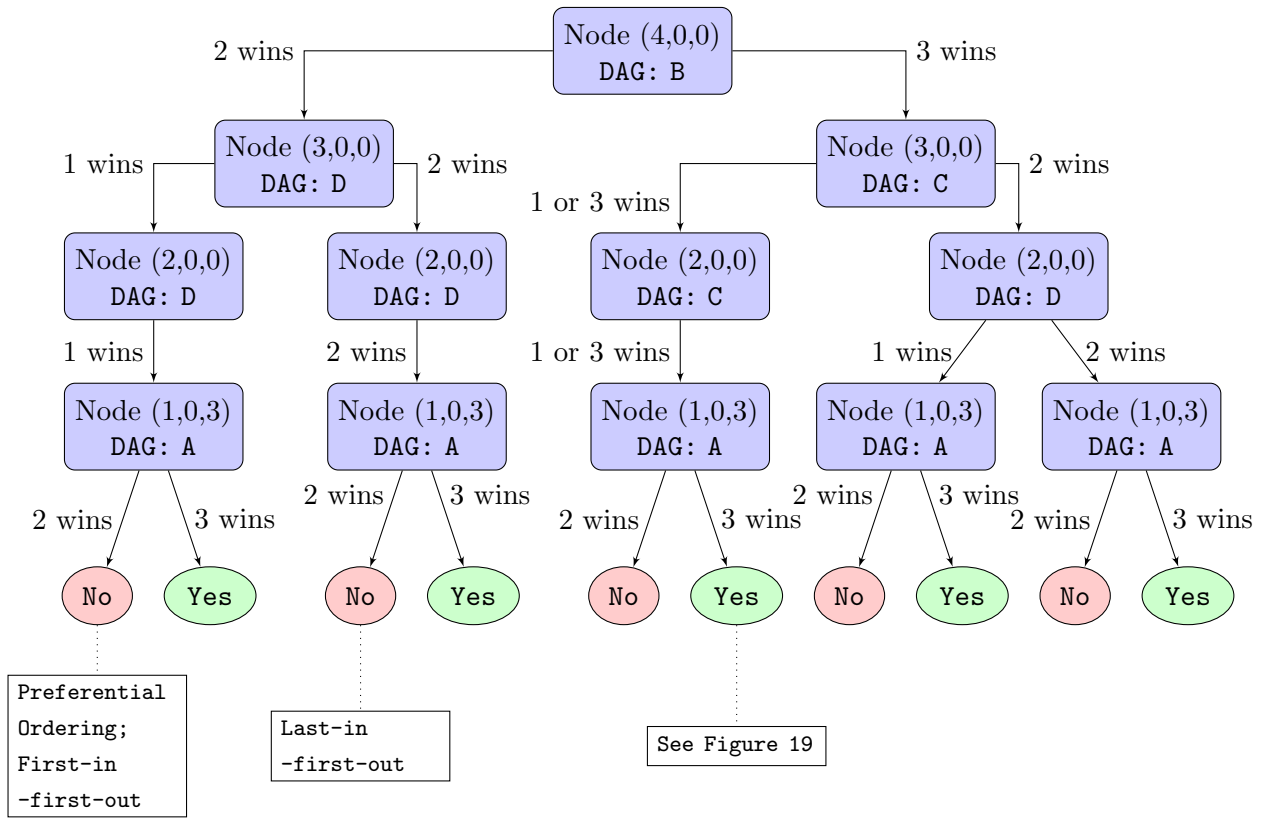

The choice of winner at  $(4,0,0)$ for overbidding graph $B$ is also between Buyer~$2$ and Buyer~$3$, but
in this case, the effect is more subtle. If Buyer~$2$ wins then the overbidding graph $D$ is induced at node $(3,0,0)$, whereas
if Buyer~$3$ wins then the overbidding graph $C$ is induced at $(3,0,0)$. In the former case, the overbidding graph $D$
arises at node $(2,0,0)$ regardless the choice of winner at $(3,0,0)$. In the latter case, there are three possible winners
in the overbidding graph $C$ at $(3,0,0)$. If Buyer~$1$ or Buyer~$3$ win these produce the same node valuations and
give the overbidding graph $C$ at $(2,0,0)$; if Buyer~$2$ wins this 
gives the overbidding graph $D$ at $(2,0,0)$. A decision tree showing all the possible choices is shown in Figure~\ref{fig:decision}.
The reader may verify that these are the only decisions that affect the valuations at the nodes.
Thus there are ten possible extensive forms, where {\tt Yes/No} details whether or not a monotonic price trajectory is produced.
Where the tie-breaking rules {\tt preferential-ordering}, {\tt first-in-first-out}, and {\tt last-in-first-out} fit in this decision tree 
are highlighted in the figure.

Several observations are in order. First, not all of the classes of tie-breaking rule give non-monotonic price trajectories.
An example of a tie-breaking rule with monotonic prices is shown in Figure~\ref{fig:monotonic}.
In fact, the choices made on the overbidding graphs $B,C$ and $D$ only affect valuations on
nodes off the equilibrium path. The equilibrium path itself is determined uniquely by the choice made for the 
overbidding graph $A$. If the winner there is Buyer~$2$ then the prices are non-monotonic;
if the winner there is Buyer~$3$ then the prices are monotonic.

We are now ready to complete the proof of the theorem. As we have just seen, any tie-breaking rule
can be classified into one of ten classes depending upon its choices on this sequential auction.
Five of the classes lead to non-monotonic prices on this instance. For the other five classes of tie-breaking rule we
need to construct different examples on which they induces non-monotonic prices.
But this is easy to do! Take exactly the same example but with the labels of Buyer~$2$ and Buyer~$3$ interchanged.
The equilibrium paths for this sequential auction using any rule in the other five classes will then have
non-monotonic price trajectories.
\end{proof}

\subsection{Negative Utilities and Overbidding}\label{sec:negative}

We now discuss a couple of interesting observations that arise from this specific sequential auction.
First we recall another interesting property of two-buyer sequential auctions: in each round of the 
auction each buyer has a non-negative value for winning the item over the other agent \cite{GS01}.
Interestingly, this property also fails to hold for multi-buyer sequential auctions! 
\begin{theorem}
There are multi-buyer sequential auctions with weakly decreasing marginal valuations
that have subgames where one agent has a negative value for winning against one
other agent.
\end{theorem}
\begin{proof}
Consider again the sequential auction shown in Figure~\ref{fig:monotonic} in Appendix~B. Focus upon the auctions with
interdependent valuations corresponding to the subgames rooted at the nodes
$(0,1,0)$, $(0,1,1)$ and $(0,1,2)$. In all three cases, Buyer~$3$ has a negative value from winning over
Buyer~$2$. For example, at node $(0,1,0)$ Buyer~$3$ has a utility of $131$ for winning but a utility of $176$ if Buyer~$2$ wins.
(Note that Buyer~$3$ does have a positive value for defeating Buyer~$1$, specifically $131-48=83$.)
Of course, this also implies there are sequential auctions with weakly decreasing marginal valuations functions
where one agent has a negative value for winning the {\bf first item} over one other agent.
\end{proof}

Second, observe in Figure~\ref{fig:3-buyers} (see also Figure~\ref{fig:interdependent2}) that in the first round Buyer~3 has a value of $176-66=110$ for winning 
over Buyer~1. This far exceeds her marginal value of $44$ for obtaining one item. 
A similar property can be seen in Figures~\ref{fig:first-in-rule} and~\ref{fig:last-in-rule} of Appendix~B.
Such ``overbidding" also arises in two-buyer sequential auctions.
The reader may wonder, however, whether this
type of ``overbidding" is responsible for the generation of non-monotonic price trajectories in multi-buyer auctions.
This is not the case. To verify this we repeated all six million experiments described in Section~\ref{sec:expts} with the ascending price 
mechanism modified to exclude the possibility of a buyer bidding higher than their marginal value for their next unit of the good. 
The proportion of instances with non-monotonic price trajectories was similar (roughly 10\% less). Moreover, there are instances where such ``overbidding'' does not arise but where the prices are non-monotonic.

\section{Experiments}\label{sec:expts}

Our experiments were based on a dataset of over six million multi-buyer sequential auctions with non-increasing valuation functions 
randomly generated from different natural discrete probability distributions. Our goal was to observe the proportion of non-monotonic price trajectories 
in these sequential auctions and to see how this varied with (i) the number of buyers, (ii) the number of items, (iii) the distribution of valuation functions, and 
(iv) the tie-breaking rule. To do this, for each auction, we computed the subgame 
perfect equilibrium corresponding to the dropout bids and evaluated the prices along the equilibrium path to test for non-monotonicity.
We repeated this test for each of the three tie breaking rules described in Section~\ref{sec:tie-breaking-1}, 
namely {\tt preferential-ordering}, {\tt first-in-first-out} and {\tt last-in-first-out}. The results from our 6,240,000 randomly generated sequential auctions are shown in Figure~\ref{fig:bar-all} in Appendix~C.

\subsection{Dataset Generation}\label{sec:exp-generation}
We now describe the methods used to generate our auction dataset. Our generator was 
implemented in \texttt{C++11}, using the GNU Compiler Collection (version 5.1.0) and the standard random number library. 
The random number library provides classes that generate pseudo-random numbers. These classes include both uniform 
random bit generators (URBGs), which generate integer sequences with a uniform distribution, and random number 
distributions, which convert the output of a URBG into various statistical distributions (such as uniform, binomial or 
Poisson distributions). In our experiments, we used the MT19937-64 implementation of the widely-used Mersenne 
Twister URBG \cite{MN98} along with the standard {\tt uniform\_int\_distribution}, {\tt binomial\_distribution} 
and {\tt poisson\_distribution} classes to generate the valuation functions in the dataset. We restricted our attention to 
integral non-increasing marginal valuations and bounded the maximum marginal value of a single item by 100. 
The purpose of this choice of constraints was to allow for a wide variety of auction instances whilst still allowing for 
a reasonable chance for ties to arise in the ascending auction mechanism, thus enabling us 
to observe any potential effects of varying the tie-breaking rule.

Our dataset contains auctions with $n=3$, $4$ and $5$ different buyers. For the 3-buyer case, we varied the 
number $T$ of items from $T=2$ to $T=16$. For each auction size, we generated a total of $240,000$ instances. 
Specifically, let $V_i$ be the valuation function of buyer~$i$, that is, $V_i(0) = 0$ and $V_i$ 
assigns a non-negative value for every integer $\ell$, $1 \leq \ell \leq T$, corresponding to buyer~$i$'s value for any 
set of $\ell$ items. Let $v_i(\ell)$ be the marginal value of buyer~$i$ for winning an $\ell^{\text{th}}$ item, that 
is, $v_i(\ell) = V_i(\ell) - V_i(\ell-1)$. We used the three aforementioned distributions to each generate $80,000$ sets of valuation functions.
To generate the values $v_i(\ell)$, for each buyer~$i$ we first uniformly 
selected a maximum marginal value $w_i$ in the interval $[1,100]$. For half the instances, we generated $w_i$ independently for every player, and for the other half, we made $w_i$ equal for all players. Subsequently, for each 
buyer~$i$, we chose the number of nonzero valuations $m_i$ uniformly in $[1,T]$. For the first distribution, we then 
independently generated and sorted $m_i$ values uniformly in $[1,w_i]$, and padded this sequence with $T-m_i$ 
zeros, to generate a decreasing integer sequence: the valuation function for buyer~$i$. For the second distribution, for each 
buyer~$i$ we generated $m_i$ values from a binomial distribution with parameters $n=w_i$ and $p=0.5$, and sorted and padded this sequence with $T-m_i$ 
zeros. For the third 
distribution, we let $u_{i,1} = w_i$, and for each $j \geq 2$ we let $u_{i,j} = \max(0,u_{i,j-1}-x_j)$, where $x_j$ was drawn 
from a Poisson distribution with parameter $\lambda = \frac{w_i}{m_i}$. 

The above steps were repeated for the 4-buyer and 5-buyer cases, varying the number of items from $T=2$ to $T=12$ in each case. In each of these cases, we generated a total of $120,000$ instances for each auction size, with $40,000$ instances generated from each of the three distributions.

Let us comment on the reasoning behind these choices for the sizes of our auctions. 
Our sequential auctions, with at most five bidders and at most ten items, could be solved extremely quickly; this allowed us to analyze our
large dataset. But it can be shown that the number of nodes in the extensive form for a sequential auction with $n$ buyers and $T$ items is exactly
\begin{equation*}
    \sum_{t=0}^T \binom{t+n-1}{n-1} = \binom{T+n}{n} =  \binom{T+n}{T}
\end{equation*}
Thus the size of the extensive form grows exponentially in the number of buyers and the number of items.
So, whilst with additional time we can easily solve slightly larger instances, we cannot expect to solve significantly larger instances. 
We remark that the auction sizes we can solve are comparable to many of the real sequential auctions described in the Introduction.

\subsection{Experimental Results}\label{sec:exp-results}
The results from our 6,240,000 randomly generated sequential auctions are shown in Figure~\ref{fig:bar-all}.
In these bar charts there is one bar for each combination of auction size and data structure ({\tt preferential-ordering}, 
{\tt first-in-first-out} and {\tt last-in-first-out}). Each bar shows the number of auctions of that type that induced 
non-monotonic prices. For three buyers there were 240,000 sequential auctions generated of each type.
For example, for sequential auctions with three buyers and five items that use the preferential-ordering tie-breaking rule,
there were 7 auctions out of 240,000 that had non-monotonic prices. For four and five buyers there
were 120,000 auctions of each type.

Figure~\ref{fig:bar-all} shows all the tests together. Recall that the valuation functions in each sequential auction 
were generated in one of three different ways (uniform, Poisson, binomial). In Appendix~C, we show in Figure~\ref{fig:threebuyers}, Figure~\ref{fig:fourbuyers} and Figure~\ref{fig:fivebuyers} these three cases for 3-buyer, 4-buyer and 5-buyer sequential
auctions respectively. We found no examples with less than 5 items that showed non-monotonicity, so the cases $T = 2,3,4$ are omitted.

As can be observed, for a fixed number of buyers, there is a slight upward drift in the proportion of non-monotonic
price trajectories as the number of items increases. Intuitively that seems unsurprising, as with longer price sequences there
are more time periods at which deviations from monotonicity can arise. A very interesting question would be to study
the limit of the proportion of non-monotonic price trajectories as the number of items gets very large. Unfortunately,
due to the exponential explosion in the number of game tree nodes discussed above, this is a question that cannot be 
studied experimentally.

\pgfplotstableread[row sep=\\,col sep=&]{
    T &   PQ &    Q &    S & Total \\
   5 &   3 &   3 &   6 &    12 \\
   6 &   0 &   0 &   6 &     6 \\
   7 &  20 &  19 &  23 &    62 \\
   8 &  20 &  25 &  26 &    71 \\
   9 &  28 &  30 &  34 &    92 \\
  10 &  26 &  27 &  31 &    84 \\
  11 &  32 &  35 &  39 &   106 \\
  12 &  53 &  54 &  58 &   165 \\
  13 &  46 &  47 &  53 &   146 \\
  14 &  56 &  61 &  60 &   177 \\
  15 &  55 &  54 &  66 &   175 \\
  16 &  38 &  41 &  48 &   127 \\
}\datathreebinom

\pgfplotstableread[row sep=\\,col sep=&]{
    T &   PQ &    Q &    S & Total \\
   5 &   2 &   2 &   2 &     6 \\
   6 &   5 &   8 &  10 &    23 \\
   7 &   8 &   8 &  12 &    28 \\
   8 &   8 &   9 &  13 &    30 \\
   9 &  11 &  12 &  14 &    37 \\
  10 &  22 &  22 &  23 &    67 \\
  11 &  22 &  18 &  32 &    72 \\
  12 &  21 &  20 &  25 &    66 \\
  13 &  16 &  21 &  26 &    63 \\
  14 &  37 &  35 &  46 &   118 \\
  15 &  37 &  34 &  41 &   112 \\
  16 &  27 &  29 &  37 &    93 \\
}\datathreeindep

\pgfplotstableread[row sep=\\,col sep=&]{
    T &   PQ &    Q &    S & Total \\
   5 &   2 &   2 &   2 &     6 \\
   6 &   2 &   3 &   6 &    11 \\
   7 &   3 &   3 &   3 &     9 \\
   8 &   8 &   6 &   9 &    23 \\
   9 &  16 &  15 &  23 &    54 \\
  10 &  21 &  17 &  31 &    69 \\
  11 &  23 &  26 &  21 &    70 \\
  12 &  24 &  19 &  28 &    71 \\
  13 &  24 &  20 &  28 &    72 \\
  14 &  26 &  24 &  27 &    77 \\
  15 &  31 &  20 &  34 &    85 \\
  16 &  22 &  19 &  28 &    69 \\
}\datathreepoiss

\pgfplotstableread[row sep=\\,col sep=&]{
    T &   PQ &    Q &    S & Total \\
   5 &   7 &   7 &  10 &    24 \\
   6 &   7 &  11 &  22 &    40 \\
   7 &  31 &  30 &  38 &    99 \\
   8 &  36 &  40 &  48 &   124 \\
   9 &  55 &  57 &  71 &   183 \\
  10 &  69 &  66 &  85 &   220 \\
  11 &  77 &  79 &  92 &   248 \\
  12 &  98 &  93 & 111 &   302 \\
  13 &  86 &  88 & 107 &   281 \\
  14 & 119 & 120 & 133 &   372 \\
  15 & 123 & 108 & 141 &   372 \\
  16 &  87 &  89 & 113 &   289 \\
}\datathreetotal

\pgfplotstableread[row sep=\\,col sep=&]{
    T &   PQ &    Q &    S & Total \\
   5 &   0 &   0 &   0 &     0 \\
   6 &   0 &   2 &   1 &     3 \\
   7 &   4 &   4 &   6 &    14 \\
   8 &   5 &   6 &   8 &    19 \\
   9 &   2 &   2 &   2 &     6 \\
  10 &   3 &   3 &  10 &    16 \\
  11 &   6 &   8 &  12 &    26 \\
  12 &  10 &  13 &  15 &    38 \\
}\datafourbinom

\pgfplotstableread[row sep=\\,col sep=&]{
    T &   PQ &    Q &    S & Total \\
   5 &   0 &   0 &   0 &     0 \\
   6 &   1 &   1 &   1 &     3 \\
   7 &   3 &   2 &   5 &    10 \\
   8 &   7 &   6 &   9 &    22 \\
   9 &  11 &  11 &  14 &    36 \\
  10 &  14 &  15 &  16 &    45 \\
  11 &  14 &  15 &  19 &    48 \\
  12 &  15 &  10 &  20 &    45 \\
}\datafourindep

\pgfplotstableread[row sep=\\,col sep=&]{
    T &   PQ &    Q &    S & Total \\
   5 &   0 &   0 &   0 &     0 \\
   6 &   0 &   1 &   1 &     2 \\
   7 &   2 &   2 &   5 &     9 \\
   8 &   4 &   3 &   4 &    11 \\
   9 &   7 &   6 &  12 &    25 \\
  10 &  12 &  11 &  20 &    43 \\
  11 &  10 &   5 &  22 &    37 \\
  12 &   7 &   6 &  13 &    26 \\
}\datafourpoiss

\pgfplotstableread[row sep=\\,col sep=&]{
    T &   PQ &    Q &    S & Total \\
   5 &   0 &   0 &   0 &     0 \\
   6 &   1 &   4 &   3 &     8 \\
   7 &   9 &   8 &  16 &    33 \\
   8 &  16 &  15 &  21 &    52 \\
   9 &  20 &  19 &  28 &    67 \\
  10 &  29 &  29 &  46 &   104 \\
  11 &  30 &  28 &  53 &   111 \\
  12 &  32 &  29 &  48 &   109 \\
}\datafourtotal

\pgfplotstableread[row sep=\\,col sep=&]{
    T &   PQ &    Q &    S & Total \\
   5 &   1 &   2 &   1 &     4 \\
   6 &   0 &   0 &   3 &     3 \\
   7 &   2 &   0 &   5 &     7 \\
   8 &   4 &   5 &   6 &    15 \\
   9 &   2 &   3 &   7 &    12 \\
  10 &   4 &   3 &   7 &    14 \\
  11 &   2 &   2 &   4 &     8 \\
  12 &   4 &   6 &   7 &    17 \\
}\datafivebinom

\pgfplotstableread[row sep=\\,col sep=&]{
    T &   PQ &    Q &    S & Total \\
   5 &   0 &   0 &   0 &     0 \\
   6 &   2 &   3 &   1 &     6 \\
   7 &   2 &   1 &   7 &    10 \\
   8 &   3 &   2 &   5 &    10 \\
   9 &   2 &   2 &   9 &    13 \\
  10 &   8 &   6 &  11 &    25 \\
  11 &   8 &   6 &  13 &    27 \\
  12 &  10 &   5 &  18 &    33 \\
}\datafiveindep

\pgfplotstableread[row sep=\\,col sep=&]{
    T &   PQ &    Q &    S & Total \\
   5 &   0 &   0 &   0 &     0 \\
   6 &   0 &   0 &   0 &     0 \\
   7 &   0 &   1 &   1 &     2 \\
   8 &   2 &   2 &   4 &     8 \\
   9 &   4 &   5 &   4 &    13 \\
  10 &   3 &   7 &   7 &    17 \\
  11 &   6 &   2 &  12 &    20 \\
  12 &   9 &   3 &  16 &    28 \\
}\datafivepoiss

\pgfplotstableread[row sep=\\,col sep=&]{
    T &   PQ &    Q &    S & Total \\
   5 &   1 &   2 &   1 &     4 \\
   6 &   2 &   3 &   4 &     9 \\
   7 &   4 &   2 &  13 &    19 \\
   8 &   9 &   9 &  15 &    33 \\
   9 &   8 &  10 &  20 &    38 \\
  10 &  15 &  16 &  25 &    56 \\
  11 &  16 &  10 &  29 &    55 \\
  12 &  23 &  14 &  41 &    78 \\
}\datafivetotal

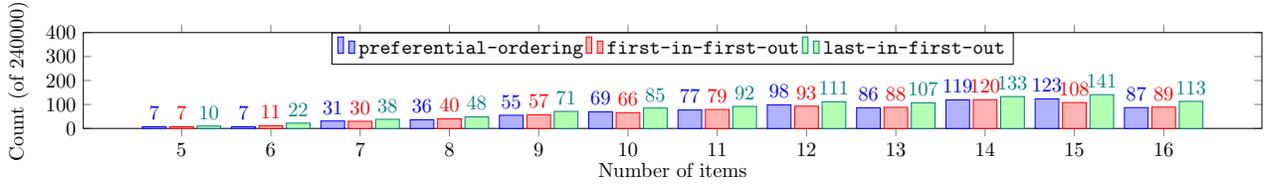
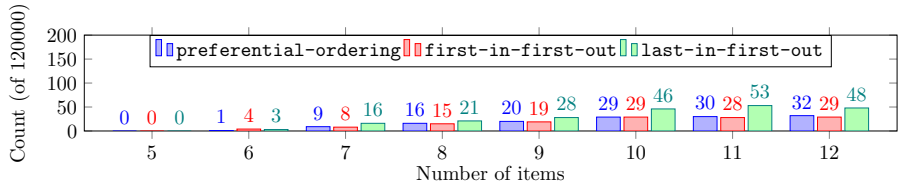
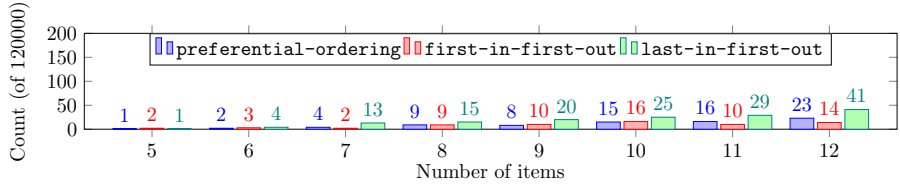
\begin{figure}[ht]
    \centering
    \begin{subfigure}{\textwidth}
        \centering
        \begin{tikzpicture}[scale=0.7]
            \begin{axis}[
                    ybar,
                    bar width=.45cm,
                    width=\textwidth,
                    height=.2\textwidth,
                    legend style={at={(0.5,1)},
                        anchor=north,legend columns=-1},
                    symbolic x coords={5,6,7,8,9,10,11,12,13,14,15,16},
                    xtick=data,
                    xlabel={Number of items},
                    nodes near coords,
                    nodes near coords align={vertical},
                    ymin=0,ymax=400,
                    ylabel={Count (of 240000)},
                    x post scale=1.45
                ]
                \addplot table[x=T,y=PQ]{\datathreetotal};
                \addplot table[x=T,y=Q]{\datathreetotal};
                \addplot[green!50!blue,fill=green!30] table[x=T,y=S]{\datathreetotal};
                \legend{{\tt preferential-ordering}, {\tt first-in-first-out}, {\tt last-in-first-out}}
            \end{axis}
        \end{tikzpicture}
        \caption{3 buyers}
        \label{fig:threetotal}
    \end{subfigure}
    
    \begin{subfigure}{\textwidth}
        \centering
        \begin{tikzpicture}[scale=0.7]
            \begin{axis}[
                    ybar,
                    bar width=.45cm,
                    width=\textwidth,
                    height=.2\textwidth,
                    legend style={at={(0.5,1)},
                        anchor=north,legend columns=-1},
                    symbolic x coords={5,6,7,8,9,10,11,12},
                    xtick=data,
                    xlabel={Number of items},
                    nodes near coords,
                    nodes near coords align={vertical},
                    ymin=0,ymax=200,
                    ylabel={Count (of 120000)},
                ]
                \addplot table[x=T,y=PQ]{\datafourtotal};
                \addplot table[x=T,y=Q]{\datafourtotal};
                \addplot[green!50!blue,fill=green!30] table[x=T,y=S]{\datafourtotal};
                \legend{{\tt preferential-ordering}, {\tt first-in-first-out}, {\tt last-in-first-out}}
            \end{axis}
        \end{tikzpicture}
        \caption{4 buyers}
        \label{fig:fourtotal}
    \end{subfigure}
    
    \begin{subfigure}{\textwidth}
        \centering
        \begin{tikzpicture}[scale=0.7]
            \begin{axis}[
                    ybar,
                    bar width=.45cm,
                    width=\textwidth,
                    height=.2\textwidth,
                    legend style={at={(0.5,1)},
                        anchor=north,legend columns=-1},
                    symbolic x coords={5,6,7,8,9,10,11,12},
                    xtick=data,
                    xlabel={Number of items},
                    nodes near coords,
                    nodes near coords align={vertical},
                    ymin=0,ymax=200,
                    ylabel={Count (of 120000)},
                ]
                \addplot table[x=T,y=PQ]{\datafivetotal};
                \addplot table[x=T,y=Q]{\datafivetotal};
                \addplot[green!50!blue,fill=green!30] table[x=T,y=S]{\datafivetotal};
                \legend{{\tt preferential-ordering}, {\tt first-in-first-out}, {\tt last-in-first-out}}
            \end{axis}
        \end{tikzpicture}
        \caption{5 buyers}
        \label{fig:fivetotal}
    \end{subfigure}
  \caption{Bar Charts showing the Frequency of Non-Monotonic Price Trajectories}
    \label{fig:bar-all}
\end{figure}

The main conclusion to be drawn from these experiments is that non-monotonic prices are extremely rare. Of the 6,240,000 auctions, the {\tt preferential-ordering} tie-breaking rule produced just 1,100 violations of the declining price anomaly. The {\tt first-in-first-out} rule gave 986 violations and the {\tt last-in-first-out} rule gave 1,334 violations. The overall observed rate of non-monotonicity over these 18 million tests was $0.000183$.

\section{Conclusion}\label{sec:conclusion}
We have shown that the declining price anomaly does not always hold in multi-buyer sequential auctions.
This result provides an explanation for the difficulty hitherto of obtaining quantitative analyses of
multi-buyer sequential auctions. Our experiments show that the declining price anomaly is very rarely violated in
randomly generated sequential auctions. 

\setlength{\bibsep}{2pt}
\bibliographystyle{apa}
\bibliography{references_monotonic}

\newpage
\section*{Appendix A}
In this appendix we prove the claim that in a second-price auction with interdependent valuations
it may be the case that no strategies survive the iterative deletion of weakly dominated strategies
if we allow domination to be by a higher bid.

\begin{theorem}
There are second-price auctions with interdependent valuations where no  strategies survive 
the iterative deletion of weakly dominated strategies
\end{theorem}
\begin{proof}
Consider the $3$-buyer auction with interdependent valuations shown in Figure~\ref{fig:nothing}.

 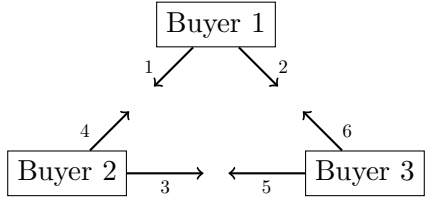
\begin{figure}[h]
\centering
 \begin{tikzpicture}[scale=0.33]
 \node[rectangle,draw](B1) at (0,0) {Buyer 1};
\node[rectangle,draw, scale=1](B2) at (-6,-6) {Buyer 2};
\node[rectangle,draw, scale=1](B3) at (6,-6) {Buyer 3};

\draw [->,  thick] (B1) -- (-2.5,-2.5) node[midway, left, scale = .66]{$1\ \ $};
\draw [->, thick] (B2) -- ( -3.5 ,-3.5) node[midway, left, scale = .66]{$4\ \ $};
\draw [->, thick] (B1) -- (2.5, -2.5) node[midway,right, scale = .66]{$\ \ 2$};
\draw [->,  thick] (B2) -- (-.5,-6) node[midway, below, scale = .66]{$3$};
\draw [->, thick] (B3) -- ( .5 ,-6) node[midway,below, scale = .66]{$5$};
\draw [->, thick] (B3) -- (3.5,-3.5) node[midway,right, scale = .66]{$\ \ 6$};
 \end{tikzpicture}
\caption{{\em A second-price auction with interdependent valuations where no strategies survive the iterative 
deletion of weakly dominated strategies.}}
\label{fig:nothing}
\end{figure}

Let's now examine what happens when we use the two different orderings  illustrated in Figure~\ref{tab:1} to delete weakly dominated strategies.

\begin{figure}[h]
	\centering
	\resizebox{13cm}{!}{
	\begin{tabular}{|c|c|c|c|}
		\hline 
		& Buyer 1 & Buyer 2 & Buyer 3 \\ 
		\hline 
		$S^0$ & $[0,\infty)$ & $[0,\infty)$ & $[0,\infty)$ \\ 
		\hline 
		$S^1$ & $[0,\infty)$ & $[0,\infty)$ & $[0,5]\setminus [1,2]$ \\ 
		\hline 
		$S^2$ & $\{1\}$ & $[0,\infty)$ & $[0,5]\setminus [1,2]$ \\ 
		\hline
		\end{tabular}
	\quad
	\begin{tabular}{|c|c|c|c|}
		\hline 
		& Buyer 1 & Buyer 2 & Buyer 3 \\
		\hline 
		$S^0$ & $[0,\infty)$ & $[0,\infty)$ & $[0,\infty)$ \\ 
		\hline 
		$S^1$ & $[0,\infty)$ & $[0,4]\setminus [1,2]$ & $[0,\infty)$ \\ 
		\hline 
		$S^2$ & $\{2\}$ & $[0,4]\setminus [1,2]$ & $[0,\infty)$ \\ 
		\hline
	\end{tabular}}
	\caption{{\em Two processes that together eliminate every strategy for Buyer~$1$.}}
	\label{tab:1}
\end{figure}
Consider first the iterative process on the LHS of Figure~\ref{tab:1}.
Observe that Buyer~$3$ is willing to pay $6$ to beat Buyer~$1$ and $5$ to beat Buyer~$2$.
It follows that any bid above $6$ is weakly dominated by a bid of $6$. Moreover, as this is a second-price auction,
any bid below $5$ is weakly dominated by a bid of $5$ (we emphasize that this latter fact does not hold 
in the case of first-price auctions).
Now, rather than deleting all these bids immediately, imagine that Buyer~$3$
deletes any bid over $6$ {\bf and} any bid between $1$ and $2$.
Therefore $S^1_3 = [0,6] \setminus [1,2]$. At this stage the undeleted strategies for Buyer~$1$ and Buyer~$2$
remain $S^1_1=S^1_2=[0,\infty)$. 

Next observe that  Buyer~$1$ is willing to pay at most $1$ to beat Buyer~$2$ but up to $2$ to beat Buyer~$3$. 
Because this is a second-price auction, it immediately follows that any bid below one or above two is weakly dominated.
Now let's compare the outcomes for Buyer~$1$ between 
bidding $1$ and bidding $x \in (1,2]$.
If Buyer~$2$ and Buyer~$3$ are both bidding below one then Buyer~$1$ wins with a bid of $1$ or a bid of $x$.
If either Buyer~$2$ or Buyer~$3$ is bidding greater than $x$ then Buyer~$1$ loses with a bid of $1$ or a bid of $x$.
So suppose the highest bid from Buyer~$2$ and Buyer~$3$ is between $1$ and $x$. If this highest bid is from Buyer~$2$ then 
Buyer~$1$ would prefer to lose and so bidding $1$ is preferable to bidding $x$. 
On the other hand, if this highest bid is from Buyer~$3$ then 
Buyer~$1$ would prefer to lose and so bidding $x$ is preferable to bidding $1$. But the latter case cannot happen as
the strategy space of Buyer~$3$ is currently $S^1_3 = [0,5] \setminus [1,2]$. It follows that bidding $1$ weakly dominates bidding $x$.
Hence we can set $S^2_1=\{1\}$.

Next consider the iterative process on the RHS of Figure~\ref{tab:1}.
This time let's begin by deleting strategies of Buyer~$2$ that are weakly dominated.
Observe that Buyer~$2$ is willing to pay $4$ to beat Buyer~$1$ and $3$ to beat Buyer~$3$.
Let's imagine that Buyer~$2$ now deletes any bid over $4$ {\bf and} any bid between $1$ and $2$.
Therefore $S^1_2= [0,4] \setminus [1,2]$. At this stage the undeleted strategies for Buyer~$1$ and Buyer~$3$
remain $S^1_1=S^1_3=[0,\infty)$. 

In the next step consider Buyer~$1$. 
Again, bidding less than one or above two is weakly dominated. This time let's compare the outcomes for Buyer~$1$ between 
bidding $2$ and bidding $x \in [1,2)$.
If Buyer~$2$ and Buyer~$3$ are both bidding below $x$ then Buyer~$1$ wins with a bid of $2$ or a bid of $x$.
If either Buyer~$2$ or Buyer~$3$ is bidding greater than $2$ then Buyer~$1$  loses with a bid of $2$ or a bid of $x$.
So suppose the highest bid from Buyer~$2$ and Buyer~$3$ is between $x$ and $2$. If this highest bid is from Buyer~$2$ then 
Buyer~$1$ would prefer to lose and so bidding $x$ is preferable to bidding $2$. 
On the other hand, if this highest bid is from Buyer~$3$ then 
Buyer~$1$ would prefer to win and so bidding $2$ is preferable to bidding $x$. But the former case cannot happen as
the strategy space of Buyer~$2$ is currently $S^1_2 = [0,4] \setminus [1,2]$. If follows that bidding $2$ weakly dominates bidding $x$.
Hence we can set $S^2_3=\{2\}$.

But $\{1\}\cap\{2\} = \emptyset$. Therefore, {\bf no} strategy for Buyer~$1$ survives the iterative deletion of weakly-dominated strategies.
That is, for each bid value $b$ there is a sequence of iterative deletions of weakly-dominated strategies that deletes
the bid $b$ of Buyer~$1$. 

\begin{figure}[h]
	\centering
	\resizebox{13cm}{!}{
\begin{tabular}{|c|c|c|c|}
		\hline 
		& Buyer 1 & Buyer 2 & Buyer 3 \\ 
		\hline 
		$S^0$ & $[0,\infty)$ & $[0,\infty)$ & $[0,\infty)$ \\ 
		\hline 
		$S^1$ & $[0,\infty)$ & $[0,\infty)$ & $[0,5]\setminus [3,4]$ \\ 
		\hline 
		$S^2$ & $[0,\infty)$ & $\{4\}$ & $[0,5]\setminus [3,4]$ \\ 
		\hline
		\end{tabular}
	\quad
	\begin{tabular}{|c|c|c|c|}
		\hline 
		& Buyer 1 & Buyer 2 & Buyer 3 \\
		\hline 
		$S^0$ & $[0,\infty)$ & $[0,\infty)$ & $[0,\infty)$ \\ 
		\hline 
		$S^1$ & $[0,2]$ & $[0,\infty)$ & $[0,\infty)$ \\ 
		\hline 
		$S^2$ & $[0,2]$ & $\{3\}$ & $[0,\infty)$ \\ 
		\hline
	\end{tabular}}
	\caption{{\em Two processes that together eliminate every strategy for Buyer~$2$.}}
	\label{tab:2}
\end{figure}

We remark that, for this example, similar arguments show that no bid for Buyer~$2$ survives the iterative 
deletion of weakly-dominated strategies. In particular, the two processes shown in Figure~\ref{tab:2}
lead to non-intersecting, undeleted, strategy sets for Buyer~$2$.
\end{proof}

\newpage
\section*{Appendix B}
In this appendix, we present the three remaining extensive forms referenced in the main text.

 \begin{figure}[!h]
\centering
 %
\caption{Non-Monotonic Prices with the {\tt first-in-first-out} Rule.}\label{fig:first-in-rule}
\end{figure}

\newpage

 \begin{figure}[!h]
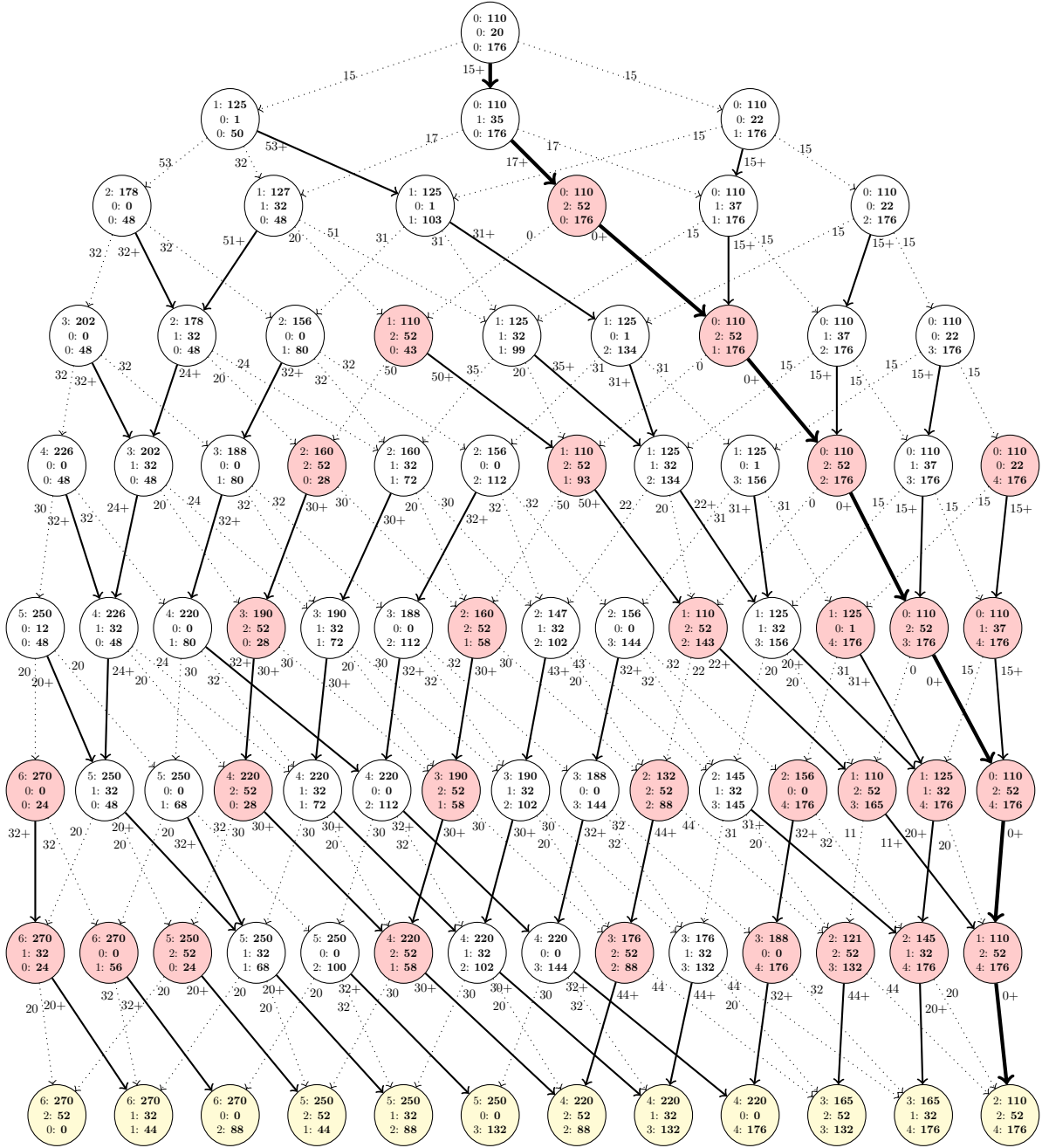

\centering
 %
\caption{Non-Monotonic Prices with the {\tt last-in-first-out} Rule.}\label{fig:last-in-rule}
\end{figure}

\newpage

 \begin{figure}[!h]
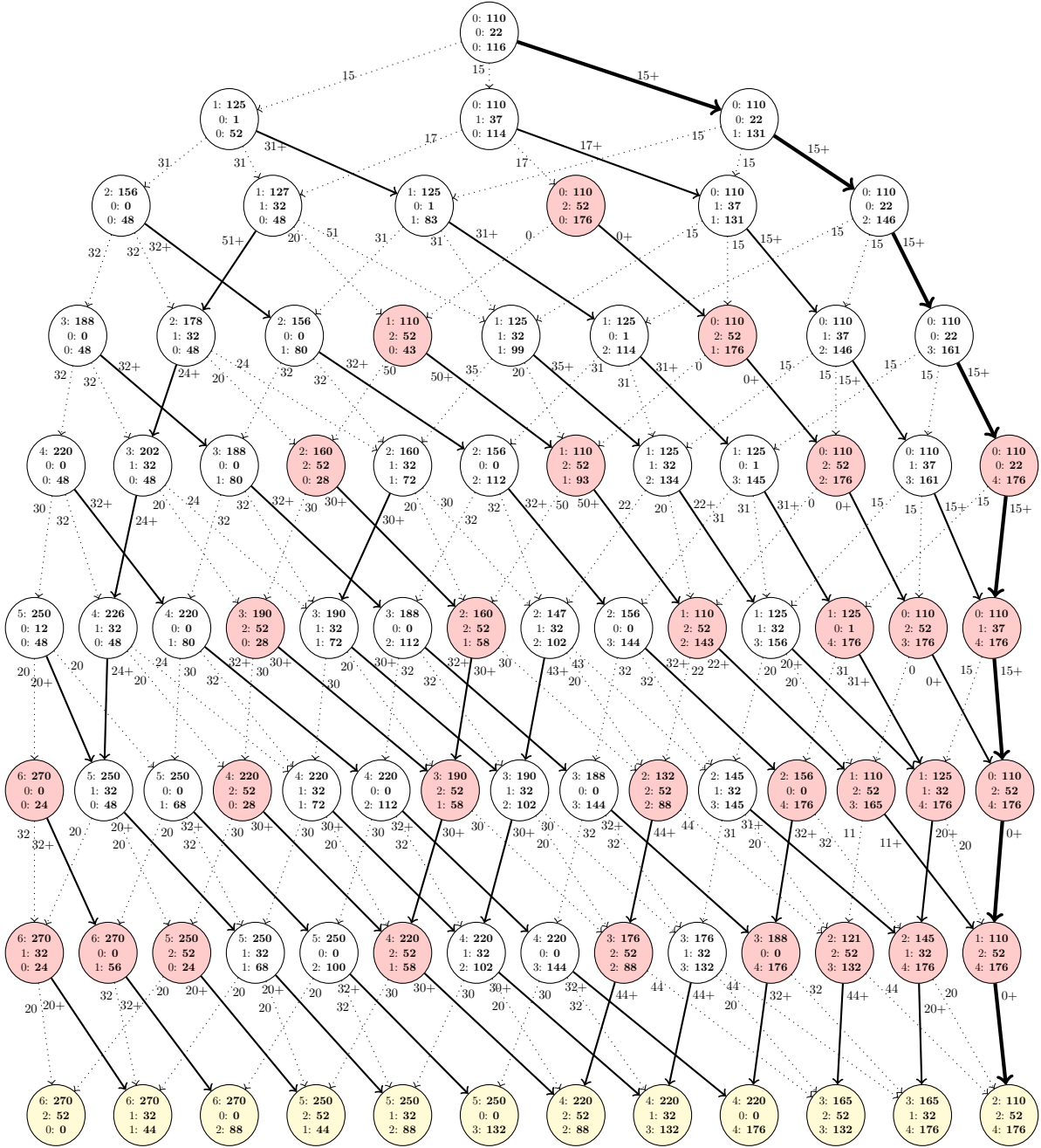

\centering
 %
\caption{A Tie-Breaking Rule with Monotonic Prices.}\label{fig:monotonic}
\end{figure}

\newpage
\section*{Appendix C}
In this appendix, we present the results for the three distributions (uniform, Poisson, binomial) for 3-buyer, 4-buyer and 5-buyer sequential auctions.

\begin{figure}[h]
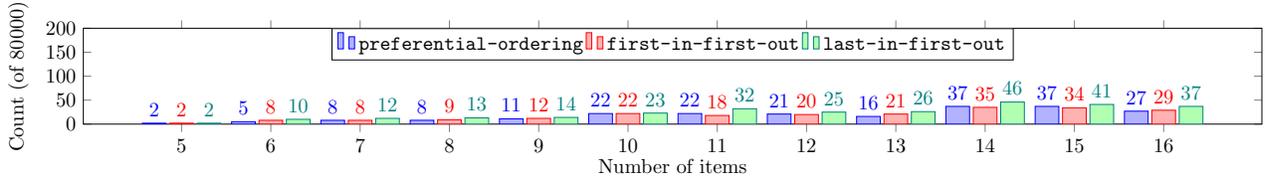
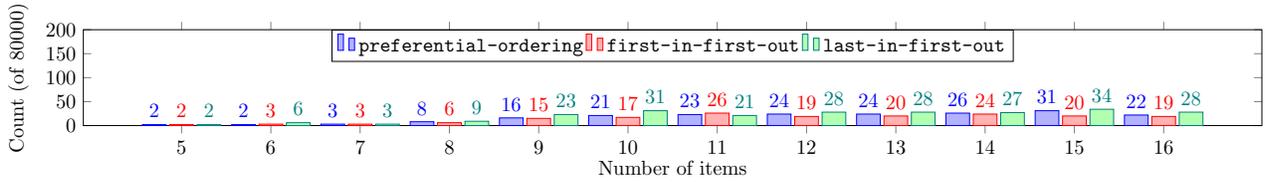
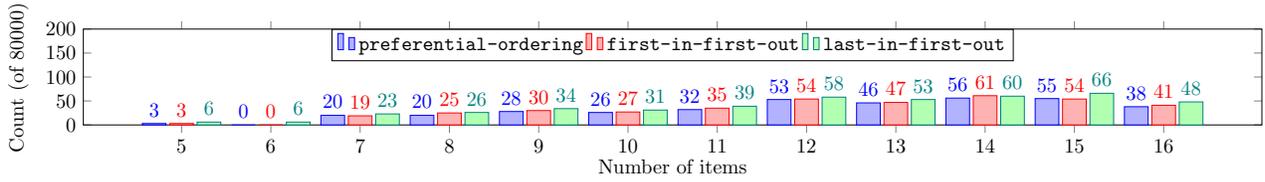

    \centering
    \begin{subfigure}{\textwidth}
        \centering
        %
        \caption{3 buyers, uniform distribution}
        \label{fig:threeuniform}
    \end{subfigure}
    
    \begin{subfigure}{\textwidth}
        \centering
        %
        \caption{3 buyers, Poisson step-size distribution}
        \label{fig:threepoisson}
    \end{subfigure}
    
    \begin{subfigure}{\textwidth}
        \centering
        %
        \caption{3 buyers, binomial distribution}
        \label{fig:threebinomial}
    \end{subfigure}
    \caption{Frequency of Non-Monotonic Price Trajectories in 3-Buyer Auctions}
\label{fig:threebuyers}
\end{figure}

\begin{figure}[h]
    \begin{subfigure}{\textwidth}
        \centering
        %
        \caption{4 buyers, uniform distribution}
        \label{fig:fouruniform}
    \end{subfigure}
    
    \begin{subfigure}{\textwidth}
        \centering
        %
        \caption{4 buyers, Poisson step-size distribution}
        \label{fig:fourpoisson}
    \end{subfigure}
    
    \begin{subfigure}{\textwidth}
        \centering
        %
        \caption{4 buyers, binomial distribution}
        \label{fig:fourbinomial}
    \end{subfigure}
        \caption{Frequency of Non-Monotonic Price Trajectories in 4-Buyer Auctions}
\label{fig:fourbuyers}
\end{figure}

\begin{figure}[h]
    \begin{subfigure}{\textwidth}
        \centering
        %
        \caption{5 buyers, uniform distribution}
        \label{fig:fiveuniform}
    \end{subfigure}
    
    \begin{subfigure}{\textwidth}
        \centering
        %
        \caption{5 buyers, Poisson step-size distribution}
        \label{fig:fivepoisson}
    \end{subfigure}
    
    \begin{subfigure}{\textwidth}
        \centering
        %
        \caption{5 buyers, binomial distribution}
        \label{fig:fivebinomial}
    \end{subfigure}
        \caption{Frequency of Non-Monotonic Price Trajectories in 5-Buyer Auctions}
\label{fig:fivebuyers}
\end{figure}

\end{document}